%% file: main.tex
\pdfoutput=1
\documentclass[12pt,draftcls,onecolumn]{IEEEtran}
\IEEEoverridecommandlockouts                              % This command is only
\ifCLASSINFOpdf
  \usepackage[pdftex]{graphicx}
  \graphicspath{{../pdf/}{../jpeg/}}
   \DeclareGraphicsExtensions{.pdf,.jpeg,.png}
\else
  \usepackage[dvips]{graphicx}
   \graphicspath{{../eps/}}
   \DeclareGraphicsExtensions{.eps}
\fi

\usepackage{epsfig} % for postscript graphics files

\usepackage{setspace}
\usepackage{amsmath,amssymb,array,graphicx,verbatim}

\usepackage{amsthm}

\usepackage{enumerate}
\usepackage[shortlabels]{enumitem}
\newcommand{\mylabel}[2]{#1#2}

\newtheorem{theorem}{Theorem}
\newtheorem{lemma}{Lemma}
\newtheorem{problem}{Problem}

\newtheorem{corollary}{Corollary}

\newtheorem{remark}{Remark}
\usepackage{enumerate}
\usepackage{epstopdf}
\usepackage[noadjust]{cite}
  
\usepackage{dblfloatfix}
\usepackage{dsfont}
\usepackage{mathrsfs}
\usepackage{relsize}
\usepackage{mathtools}% Loads amsmath
\usepackage{url}

\usepackage{tikz}
\usepackage{circuitikz}
\usepackage{tikz-cd}
\usetikzlibrary{arrows,shapes,calc,positioning}
\tikzstyle{block} = [draw,rectangle, rounded corners, minimum width=1cm, minimum height=0.8cm,text centered, line width=2pt ]
\tikzstyle{arrow} = [thick,->,>=stealth,line width=2pt]

\tikzset{cross/.style={cross out, draw=black, minimum size=2*(#1-\pgflinewidth), inner sep=0pt, outer sep=0pt},
cross/.default={1pt}}

\usepackage{enumerate}
\tikzset{
  shift left/.style ={commutative diagrams/shift left={#1}},
  shift right/.style={commutative diagrams/shift right={#1}}
}
\usetikzlibrary{calc}

\newcommand\rsmraise[1]{%
  \ifx#1\displaystyle .8\else
    \ifx#1\textstyle .8\else
      \ifx#1\scriptstyle .6\else
        .45%
      \fi
    \fi
  \fi}

\usepackage{epstopdf}
\usetikzlibrary{shapes}
\tikzstyle{block} = [draw,rectangle, rounded corners, minimum width=1cm, minimum height=0.8cm,text centered, line width=2pt ]
\tikzstyle{arrow} = [thick,->,>=stealth,line width=2pt]

\usetikzlibrary{arrows,decorations.markings,decorations}
\tikzset{% define addarrow decoration
    addarrow/.style={decoration={markings, mark=at position 1 with {\arrow{stealth}}},
                     postaction={decorate}}
}

\newcommand{\tp}{\intercal}		% transpose
\newcommand{\R}{\mathbb{R}}			% real numbers
\newcommand{\ind}{\mathds{1}}		% transpose

\graphicspath{{./Sim_results/}}

\DeclareMathOperator{\ee}{\mathbb{E}}			% expected value
\DeclareMathOperator{\prob}{\mathbb{P}}			% probability
\DeclareMathOperator{\vecc}{\mathbf{vec}}		% probability
\DeclareMathOperator{\tr}{\mathbf{tr}}			% trace
\DeclareMathOperator{\cov}{\mathbf{cov}}		% covariance
\newcounter{l1}
\newcounter{l2}
\newcommand{\bdashlist}{\begin{list}{$-$}{} }
\newcommand{\bromalist}{\begin{list}{\roman{l1}.}{\usecounter{l1}}}
\newcommand{\balphlist}{\begin{list}{(\alph{l2})}{\usecounter{l2}}}

\title{
Optimal Local and Remote Controllers with Unreliable Uplink Channels
}

\author{Seyed Mohammad Asghari, Yi Ouyang, and Ashutosh Nayyar% <-this % stops a space
\thanks{Preliminary version of this paper appears in the proceedings of the 55th annual Conference on Decision and Control (CDC), 2016 (see \cite{Ouyang_Asghari_Nayyar:CDC_2016}).}
\thanks{
S. M. Asghari and A. Nayyar are with the Department of Electrical
Engineering, University of Southern California, Los Angeles, CA. Y. Ouyang is currently with the University of California, Berkeley. 
Email: asgharip@usc.edu; ouyangyi@berkeley.edu; ashutosn@usc.edu.}
 \thanks{This research was supported by NSF under grants ECCS 1509812 and CNS 1446901.}% <-this % stops a space
}

\begin{document}

\maketitle
%\thispagestyle{empty}
%\pagestyle{empty}

%%%%%%%%%%%%%%%%%%%%%%%%%%%%%%%%%%%%%%%%%%%%%%%%%%%%%%%%%%%%%%%%%%%%%%%%%%%%%%%%
\begin{abstract}

We consider a networked control system consisting of a remote controller and a collection of linear plants, each associated with a local controller. Each local controller directly observes the state of its co-located plant and can inform the remote controller of the plant's state through an unreliable uplink channel. 
We assume that the downlink channels from the remote controller to local controllers are perfect. 
The objective of the local controllers and the remote controller is to cooperatively minimize a quadratic performance cost.
We provide a dynamic program for this decentralized control problem using the common information approach. Although our problem is not a partially nested problem, we obtain explicit optimal strategies for all controllers.
 In the optimal strategies, all controllers compute common estimates of the states of the plants based on the common information obtained from the communication network. The remote controller's action is linear in the common state estimates, and the action of each local controller is linear in both the actual state of its co-located plant and the common state estimates.
We illustrate our results with numerical experiments using randomly generated models.
\end{abstract}

%%%%%%%%%%%%%%%%%%%%%%%%%%%%%%%%%%%%%%%%%%%%%%%%%%%%%%%%%%%%%%%%%%%%%%%%%%%%%%%%

\section{Introduction}
\label{sec:intro}
\input{Introduction.tex}

\section{System Model and Problem Formulation}
\label{sec:model}
\input{Model.tex}

\section{Equivalent Problem and Dynamic Program}
\label{sec:structure}

\input{DynamicProgram.tex}

\section{Optimal Control Strategies}
\label{sec:solution}
\input{OptimalStrategies.tex}

\section{Discussion}
\label{sec:discussion}
\input{Discussions.tex}

\section{Numerical Experiments}
\label{sec:numerical}
\input{NumericalExperiments.tex}

\vspace{1mm}
\section{Conclusion}
\label{sec:conclusion}
\input{Conclusion.tex}

%%%%%%%%%%%%%%%%%%%%%%%%%%%%%%%%%%%%%%%%%%%%%%%%%%%%%%%%%%%%%%%%%%%%%%%%%%%%%%%%
%\section*{Acknowledgment}
\vspace{-2mm}
\bibliographystyle{ieeetr}
\bibliography{IEEEabrv,References,packet_drop,collection}

%\newpage
\vspace{-3mm}
\appendices
\input{Appendices.tex}

\end{document}

%% file: Introduction.tex
The advent of information and communication technologies along with the development of the Internet of Things (IoT)
has drawn increasing attention to networked control systems (NCSs). NCSs are distributed systems in which information is exchanged through a network among various components (controllers, smart sensors, actuators, etc.).
The connectivity of NCS brings numerous opportunities to new applications such as autonomous vehicles, smart grid, remote surgery, smart home, and large manufacturing systems (see \cite{gupta2010networked,HespanhaSurvey,ConnectedVehicles} and references therein).
However, the network connection is subjected to various communication constraints.
One main constraint is the unreliability of communication channels which can greatly affect the performance of NCS \cite{zhang2001stability, Schenato2007}. Therefore, the study of NCS over unreliable channels is of great importance.

The effect of control over unreliable channels has been investigated in \cite{Imer2006optimal,Sinopoli2005,Sinopoli2006,Elia2004, Garone2008,Gupta_2010} for NCS with a single controller.
However, most NCS applications consist of multiple sub-systems where each sub-system may be controlled by a remote controller as well as a local controller and the overall system performance depends on the coordination among the remote controller and all local controllers through the communication network.
In this paper, we consider an NCS consisting of a remote controller and a collection of linear plants, each associated with a local controller as shown in Fig. \ref{fig:SystemModel}.
Each plant is directly controlled by a local controller which can perfectly observe the state of the plant. 
The remote controller can control all plants, but it does not have direct access to the states as its name suggests.
The objective of the local controllers and the remote controller is to cooperatively minimize an overall quadratic performance cost of the NCS. 
The remote controller and local controllers are connected by a communication network where the downlinks from the remote controller to local controllers are perfect but the uplinks from local controllers to the remote controller are unreliable channels with random packet drops. 
Such scenario happens in many situations where the remote controller is equipped with sufficient communication resources, but each local controller has limited transmission capabilities.
For instance, the local controllers can be a group of battery-powered telerobots or autonomous vehicles with limited transmission power proximal to their co-located systems while
the remote controller can be a controlling operator connected to a power outlet or a base station with high transmission power.

The NCS structure we study models various networked systems architectures: 1) The remote controller can model a global controller that affects all local dynamics. For example, in a smart building, the central AC unit plays the role of a remote controller that affects the temperature of multiple rooms which may also have local controllers. Furthermore, in many remotely controlled systems such as UAVs, certain low-level functions like collision avoidance are controlled by a local processor, but many high-level mission-related functions are remotely controlled by a ground control station \cite{seiler2001}; 2) System-wide global references and constraints could be modeled by the remote controller's actions. For example, the remote controller's action can describe the target location for a robot formation problem; and 3) The remote controller can be used to model an access point or a base station that relays and broadcasts information for all local controllers. For example, in vehicle to infrastructure (V2I) communication, the remote controller/access point can relay information among a set of autonomous vehicles  \cite{horowitz2000control}. 

When the local controllers are smart sensors or encoders that can only sense and transmit information, the NCS operation depends only on remote estimation and control. Remote estimation with a single smart sensor has been studied in \cite{LipsaMartins:2011,NayyarBasarTeneketzisVeeravalli:2013,Nourian_2014,Knorn_Dey_2015} and has been extended to the case with multiple smart sensors and general packet drop models in \cite{Gupta_Martins_2009,Gupta_Dana_2009}.
Remote estimation and control of a linear plant has been studied in \cite{Gupta_Hassibi_2007,BansalBasar:1989a,TatikondaSahaiMitter:2004,NairFagnaniZampieriEvan:2007,molin2013optimality,rabi2014separated} under various channel models between smart sensors and a remote controller.
The problem considered in this paper is different from these previous works on NCS because 
our problem is a decentralized control problem with multiple controllers where the dynamics of each plant is controlled by the remote controller as well as the corresponding local controller.
Finding optimal strategies in decentralized control problems is generally considered a difficult problem (see \cite{Witsenhausen:1968, LipsaMartins:2011b,blondel2000survey}).  In general, linear control strategies are not optimal, and even the problem of finding the best linear control strategies is not convex \cite{YukselBasar:2013}. 
Existing optimal solutions of decentralized control problems require either specific information structures, such as partially nested \cite{HoChu:1972,LamperskiDoyle:2011,LessardNayyar:2013,ShahParrilo:2013, Nayyar_Lessard_2015,Lessard_Lall_2015}, stochastically nested \cite{Yuksel:2009}, or other specific properties, such as quadratic invariance \cite{RotkowitzLall:2006} or substitutability \cite{AsghariNayyar:2015, Asghari_Nayyar_substitutability_TAC}. A two-controller partially-nested decentralized control problem with packet drop channels from controllers to actuators but with perfect one-directional communication from controller 1 to controller 2 was investigated in \cite{Chang_Lall_finite,Chang_Lall_infinite}.

For the problem we consider in this paper, none of the above properties hold either due to the unreliable inter-controller communication or due to the nature of dynamics and cost function. We use the common information approach to show that this problem is equivalent to a centralized sequential decision-making problem where the remote controller is the only decision-maker. We provide a dynamic program to obtain the optimal strategies of the remote controller in the equivalent problem. Then, using the optimal strategies of the equivalent problem, we obtain explicit optimal strategies for all local controllers and the remote controller. In the optimal strategies, all controllers compute common estimates of the states of the plants based on the common information obtained from the communication network. The remote controller's action is linear in the common state estimates, and the action of each local controller is linear in both the actual state of its co-located plant and the common state estimates.

\subsection*{Contributions of the Paper}
The main contributions of the paper are as follows.
\begin{enumerate}
\item 
We investigate a decentralized stochastic control problem in which local controllers send their information to a remote controller over unreliable links. To the best of our knowledge, this is the first paper that solves an optimal decentralized control problem with unreliable communication between controllers (in contrast to problems in networked control systems and remote estimation problems where the unreliable communication is between sensors/encoders and controller or between controllers and actuators).

\item The information structure of our problem is not partially nested, hence we cannot a priori restrict to linear strategies for optimal control. 
We use ideas from the common information approach of \cite{nayyar2013decentralized} to compute optimal controllers.
Since the state and action spaces of our problem are  Euclidean spaces, the results and arguments of \cite{nayyar2013decentralized} for finite spaces cannot be directly applied. 
We provide a complete set of results to adapt the common information approach to our linear-quadratic setting with non-partially nested information structure.
Our rigorous proofs carefully handle the issues of measurability constraints, the existence of well-defined value functions and infinite dimensional strategy spaces.

\item We show that the optimal control strategies of this problem admit simple structures--
the optimal remote control is linear in the common estimates of system states and each optimal local control is linear in both the common estimates of system states and the perfectly observed local state. 
The main strengths of our result are that (i) it provides a simple strategy that is proven to be optimal: not only is the strategy in Theorem 3 linear, it uses estimates that can be easily updated; (ii) it provides a tractable way of computing the gain matrices involved in the optimal strategy. In fact, our numerical experiments indicate that the computational burden of finding the optimal gain matrices in our decentralized problem is comparable to finding optimal strategies in a corresponding centralized LQ problem.

\item Our results apply to any noise model with zero mean and finite second moments. In fact, the optimal control strategies of our problem are independent of the noise statistics.

\end{enumerate}

\vspace{-2mm}
\subsection{Notation}
Random variables/vectors are denoted by upper case letters, their realization by the corresponding lower case letter.
For a sequence of column vectors $X, Y, Z,...$, the notation $\vecc(X,Y,Z,...)$ denotes the vector $[X^{\tp}, Y^{\tp}, Z^{\tp},...]^{\tp}$. The transpose and trace of matrix $A$ are denoted by $A^{\tp}$ and $\tr(A)$, respectively. 
In general, subscripts are used as time index while superscripts are used to index controllers.
For time indices $t_1\leq t_2$, $X_{t_1:t_2}$ (resp. $g_{t_1:t_2}(\cdot)$) is the shorthand notation for the variables $X_{t_1},X_{t_1+1},...,X_{t_2}$ (resp.  functions $g_{t_1}(\cdot),\dots,g_{t_1}(\cdot)$). Similarly, 
for $n_1 \leq n_2$, $X^{n_1:n_2}$ (resp. $g^{n_1:n_2}(\cdot)$) is the shorthand notation for the variables $X^{n_1},X^{n_1+1},...,X^{n_2}$ (resp.  functions $g^{n_1}(\cdot),\dots,g^{n_2}(\cdot)$).
For $n_1 \leq n_2$, the product of sets $E^{n_1},E^{n_1+1},...,E^{n_2}$ is denoted by $E^{n_1:n_2}$. Similarly, we use $E_{t_1:t_2}$ to denote the product of sets $E_{t_1},E_{t_1 +1}, \ldots, E_{t_2}$.
For set $\mathcal{A} = \{\alpha_1,\ldots, \alpha_N \}$, the collection of random variables $X^{\alpha_1}, \ldots, X^{\alpha_N}$
 (resp. functions $g^{\alpha_1}(\cdot),\ldots, g^{\alpha_N}(\cdot)$) is denoted by $\{X^{m}\}_{m \in \mathcal{A}}$ (resp. $\{g^m(\cdot)\}_{m \in \mathcal{A}}$).
  Furthermore, the collection of random variables $\{X^{m}\}_{m \in \mathcal{A} \setminus \{\alpha_n \}}$ (resp.  functions $\{g^m(\cdot)\}_{m \in \mathcal{A} \setminus \{\alpha_n \}}$ ) is denoted by $X^{-\alpha_n}$ (resp. $g^{-\alpha_n}$) and the product of sets $E^m$, $m \in \mathcal{A} \setminus \{\alpha_n \}$, is denoted by $E^{-\alpha_n}$. For a collection of random variables $\{X^{m}\}_{m \in \mathcal{A}}$ and  sets $E^m$, $m \in \mathcal{A}$, we use $\vecc(\{X^{m}\}_{m \in \mathcal{A}}) \in E^{\alpha_1} \times E^{\alpha_2} \times \ldots \times E^{\alpha_N}$ to denote that $X^m \in  E^m$ for all $m \in \mathcal{A}$. Moreover, the intersection of the events $E^{\alpha_1}, \ldots, E^{\alpha_N}$ is denoted by $\{E^m\}_{m \in \mathcal{A}}$. For example, $\prob(\{E^m\}_{m \in \mathcal{A}})$ denotes $\prob(\cap_{m \in \mathcal{A}}E^m\})$.

The indicator function of set $E$ is denoted by $\mathds{1}_{E}(\cdot)$, that is, $\mathds{1}_{E}(x) = 1$ if $x \in E$, and $0$ otherwise. If $E$ is an event, then $\mathds{1}_{E}$ denotes the resulting random variable.
$\prob(\cdot)$, $\ee[\cdot]$, and $\cov(\cdot)$ denote the probability of an event, the expectation of a random variable/vector, and the covariance matrix of a random vector, respectively.
For random variables/vectors $X$ and $Y$, $\prob(\cdot|Y=y)$ denotes the probability of an event given that $Y=y$, and $\ee[X|y] := \ee[X|Y=y]$. 
For a strategy $g$, we use $\prob^g(\cdot)$ (resp. $\ee^g[\cdot]$) to indicate that the probability (resp. expectation) depends on the choice of $g$. 
Let $\Delta(\R^n)$ denote the set of all probability measures on $\R^n$ with finite second moment. 
For any $\theta \in \Delta(\R^n)$, $\theta(E) = \int_{\R^n} \mathds{1}_{E}(x) \theta(dx)$ denotes the probability of event $E$ under $\theta$. 
The mean and the covariance of a distribution $\theta \in \Delta(\R^n)$ are denoted by $\mu(\theta)$ and $\cov(\theta)$, respectively, and are defined as $\mu(\theta) = \int_{R^n} x \theta(dx)$ and 
$\cov(\theta) = \int_{R^n} (x - \mu(\theta)) (x - \mu(\theta))^{\tp} \theta(dx)$.  

The notation $\mathbf{I}_{n}$ and $\mathbf{0}_{n \times m}$ is used to denoted a $n \times n$ identity matrix and a $n \times m$ zero matrix, respectively.
For block matrix $B$, $[B]_{n \bullet}$ denotes the $n$-th block row of $B$. For example, for 
$B = \begin{bmatrix}
\mathbf{I}_{m} &\mathbf{0}_{m \times n} \\ \mathbf{0}_{n \times m}  &\mathbf{I}_{n}
\end{bmatrix}$, $[B]_{1 \bullet} = [\mathbf{I}_{m}\hspace{3mm} \mathbf{0}_{m \times n}]$ and $[B]_{2 \bullet} = [\mathbf{0}_{n \times m} \hspace{3mm} \mathbf{I}_{n}]$.

\vspace{-2mm}
\subsection{Organization}
The rest of the paper is organized as follows. We introduce the system model and formulate the multi-controller NCS problem in Section \ref{sec:model}. In Section \ref{sec:structure}, we formulate an equivalent problem using the common information approach and provide a dynamic program for this problem. 
We solve the dynamic program in Section \ref{sec:solution}. 
In Section \ref{sec:discussion}, we discuss some key aspects of our approach and results. In Section \ref{sec:numerical}, we present  some numerical experiments.
Section \ref{sec:conclusion} concludes the paper. The proofs of all the technical results of the paper appear in the Appendices.

%% file: Model.tex
Consider a discrete-time system with $N$ plants, $N$ local controllers, $C^1,C^2, \ldots,C^N$, and one remote controller $C^0$ as shown in Fig. \ref{fig:SystemModel}. We use $\mathcal{N}$ to denote the set $\{1,2, \ldots, N\}$ and $\overline{\mathcal{N}}$ to denote $\{0,1,\ldots, N\}$.
The linear dynamics of plant $n \in \mathcal{N}$ are given by
\begin{align}
&X_{t+1}^n \!=\! A^{nn} X_t^n + B^{nn}U^{n}_t+ B^{n0} U^0_t + W_t^n, t=0,\dots,T,
 \label{Model:system}
\end{align}
where $X_t^n\in \R^{d_X^n}$ is the state of the plant $n \in \mathcal{N}$ at time $t$,
$U^n_t \in \R^{d_U^n}$ is the control action of the controller $C^{n}$, $n \in \overline{\mathcal{N}}$, and $A^{nn}, B^{nn}, B^{n0}$, $n \in \mathcal{N}$, are matrices with appropriate dimensions.
$X^n_0$ is a  random vector with distribution $\pi_{X_0^n}$, $W^n_t \in \R^{d_X^n}$ is a zero-mean noise vector  with distribution $\pi_{W^n_t}$.
$X_0^{1:N},W_{0:T}^{1:N}$ are independent random vectors with finite second moments.
Note that we do not assume that $X_0^{1:N}$ and $W_{0:T}^{1:N}$ are Gaussian.

The overall dynamics can be written as
\begin{align}
X_{t+1} = A X_t + BU_t + W_t
\label{overall_state_dynamic}
\end{align}
where $X_t = \vecc(X^{1:N}_t), U_t = \vecc(U^{0:N}_t),W_t = \vecc(W^{0:N}_t)$ and $A,B$ are defined as

\begin{small}
\begin{align}
A &= \begin{bmatrix}
   A^{11} & & \text{\huge0}\\
          & \ddots & \\
     \text{\huge0} & & A^{NN}
\end{bmatrix}, 
B= 
\begin{bmatrix}
B^{10} &  B^{11} & & \text{\huge0}\\
\vdots     &     & \ddots & \\
B^{N0}  &   \text{\huge0} & & B^{NN}
\end{bmatrix}.
\label{eq:thm_matricesABB}
\end{align}
\end{small}

\begin{singlespace}
\input{system_model}
\end{singlespace}
\input{timeline}

\emph{Communication Model:} At each time $t$ the local controller $C^{n}$, $n \in \mathcal{N}$, perfectly observes the state $X^n_t$ and sends the observed state to the remote controller $C^0$ through an unreliable channel with link failure probability $p^n$. 
Let $\Gamma_t^n$ be a Bernoulli random variable describing the state of this channel with $\Gamma_t^n=0$ when the link is broken and $\Gamma_t^n=1$ otherwise.
We assume that $\Gamma_{0:T}^{1:N}$ are independent  random variables and they are independent of $X_0^{1:N}$and $W_{0:T}^{1:N}$.
Furthermore, let $Z_t^n$ be the output of the channel between the local controller $C^{n}$ and the remote controller $C^0$ at time $t$. Thus, we have
 \begin{align}
\Gamma_t^n = &\left\{\begin{array}{ll}
1 & \text{ with probability }(1-p^n),\\
0 & \text{ with probability }p^n.
\end{array}\right.
\\
Z_t^n = &
\left\{\begin{array}{ll}
X_t^n & \text{ when } \Gamma_t^n = 1,\\
\emptyset & \text{ when } \Gamma_t^n = 0.
\end{array}\right.
\label{Model:channel}
\end{align}

Unlike the unreliable uplinks, we assume that there exist perfect links from $C^0$ to $C^{n}$, for $n \in \mathcal{N}$. Therefore, $C^0$ can share $Z_t^{1:N}$ and $U_{t-1}^0$ with all local controllers $C^{1:N}$.
All controllers select their control actions at time $t$ after observing $Z_t^{1:N}$. 
A schematic of the time ordering of  variables is shown in Fig. \ref{fig:timing}. 
We assume that for all $n \in \mathcal{N}$, the links from controllers $C^{n}$ and $C^0$ to the plant $n$ are perfect.

\emph{Information structure and cost:} Let $H^{n}_t$ denote the information available to controller $C^{n}$, $n \in \overline{\mathcal{N}}$, to make decisions at time $t$.
Then,
\begin{align}
H^{n}_t&= \{X^n_{0:t}, U^{n}_{0:t-1}, Z_{0:t}^{1:N}, U^0_{0:t-1}\}, \hspace{2mm} n \in \mathcal{N} \notag \\
H^0_t &= \{Z_{0:t}^{1:N}, U^0_{0:t-1}\}. 
\label{Model:info}
\end{align}
Let $\mathcal{H}^{n}_t$ be the space of all possible realizations of $H_t^n$.
Then, $C^{n}$'s actions are selected according to
\begin{align}
U^{n}_t &= g^{n}_t(H^{n}_t), \hspace{2mm}  n \in \overline{\mathcal{N}},
\label{Model:strategy}
\end{align}
where $g^{n}_t:\mathcal{H}^{n}_t  \to  \R^{d_U^n}$ is a Borel measurable mapping.
The collection of mappings $g_0^n, \ldots, g_T^n$ is called the strategy of controller $C^n$ and is denoted by $g^n$. The collection of all controllers' strategies, $g^{0:N}$, is called the strategy profile\footnote{When it is clear from the context, we will simply use $g$ instead of $g^{0:N}$.}.

The instantaneous cost $c_t(X_t^{1:N}, U^{0:N}_t)$ of the system is a general quadratic function given by
\begin{align}
&c_t(X_t^{1:N}, U^{0:N}_t) = 
S_t^\tp R_tS_t, ~\mbox{where} \notag
\\
&S_t = \vecc(X_t^{1:N},U^{0:N}_t), R_t = \left[\begin{array}{lll}
R^{XX}_t & R^{XU}_t \\
R^{UX}_t & R^{UU}_t 
\end{array}\right],
\label{cost_function}
\end{align}
and
\begin{align}
R^{XX}_t &= \begin{bmatrix}
R^{X^1X^1}_t &\ldots &R^{X^1X^N}_t, \notag \\
\vdots & \ddots & \vdots \\
R^{X^NX^1}_t & \ldots & R^{X^NX^N}_t
\end{bmatrix} =: [R^{X^i X^j}_t]_{i,j \in \mathcal{N}},  \notag \\
R^{XU}_t &= (R^{UX}_t)^{\tp} = [R^{X^i U^j}_t]_{i \in \mathcal{N}, j \in  \overline{\mathcal{N}}}, R^{UU}_t = [R^{U^iU^j}_t]_{i,j \in \overline{\mathcal{N}}}.
\label{matrix_structure}
\end{align}
$R_t$ is a symmetric positive semi-definite (PSD) matrix and $R_t^{UU}$ is a symmetric positive definite (PD) matrix. 

\emph{Problem Formulation:} The performance of strategies $g^{n}:=g^{n}_{0:T}$, $n \in \overline{\mathcal{N}}$, is measured by the total expected cost over a finite horizon $T$:
\begin{align}
J(g^{0:N})=\ee^{g^{0:N}}\left[\sum_{t=0}^T c_t(X^{1:N}_t,U^{0:N}_t)\right].
\label{Model:obj}
\end{align}

Let $\mathcal{G}^{n} $ be the set of all possible control strategies for $C^{n}$, $n \in  \overline{\mathcal{N}}$.
Then, the optimal control problem can be formally defined as follows.

\begin{problem}
\label{problem1}
For the system model described above by \eqref{Model:system}-\eqref{Model:obj}, we would like to solve the following strategy optimization problem,
\begin{align}
\inf_{g^{0} \in \mathcal{G}^{0},\ldots,g^{N} \in \mathcal{G}^{N}  } J(g^{0:N}).
\label{Model:optimization_goal}
\end{align}
\end{problem}

\begin{remark}
Without loss of optimality, we can restrict attention to strategy profiles $g^{0:N}$ that ensure a finite expected cost at each time step. Because $R_t$ is positive semi-definite and $R_t^{UU}$ is positive definite, finite expected cost at each time $t$ is equivalent to 
\begin{align}
\ee^{g^{0:N}} [(U_t^n)^{\tp} U_t^n] 
= \ee^{g^{0:N}} [g_t^n(H_t^n)^{\tp} g_t^n(H_t^n)] 
&< \infty, \notag \\
& \hspace{-1cm} \forall n \in \overline{\mathcal{N}}, \forall t.
\label{strategy_condition}
\end{align}
Therefore, in the subsequent analysis we will implicitly assume that the strategy profile under consideration, $g^{0:N}$, ensures that for all time $t$ and for all $n \in \overline{\mathcal{N}}$, $g^{n}_t:\mathcal{H}^{n}_t \to  \R^{d_U^n}$
has finite second moments, that is, \eqref{strategy_condition} holds.
\end{remark}

Problem \ref{problem1} is a  decentralized optimal control problem with $(N+1)$ controllers. Decentralized optimal control problems are generally believed to be hard because (i) linear strategies may not be globally optimal and (ii) the strategy optimization problem may be a non-convex problem over infinite dimensional spaces \cite{mahajan2015sufficient}. For decentralized linear-quadratic-Guassian (LQG) control problems with partially-nested information structure,  however, linear control strategies are known to be optimal \cite{HoChu:1972}. An information structure is partially-nested if whenever the action of a controller affects the information of another controller, the latter knows whatever the former knows.  
Note that Problem \ref{problem1} is not a  partially nested LQG problem.
In particular, $C^{n}$'s action $U^{n}_{t-1}$, $n \in \mathcal{N}$, affects $X^n_t$, 
and consequently, it affects $Z^n_t$.
Since $Z^n_t$ is a part of the remote controller $C^0$'s information $H^0_t$ at $t$ but $H^{n}_{t-1} \not \subset H^0_t$, the information structure in Problem \ref{problem1} is not partially nested.
Furthermore, in Problem \ref{problem1}, $X_0^{1:N}$ and $W_{0:T}^{1:N}$ are not necessarily Gaussian. 
Therefore, we cannot a priori assume that linear control strategies are  optimal for Problem \ref{problem1}. This means we have to optimize over the full space of control strategies rather than the finite-dimensional subspace of linear strategies.

Our approach to Problem \ref{problem1} is based on the common information approach \cite{nayyar2013decentralized} for decentralized decision-making.
We identify the common information among the $N+1$ controllers and use it to define a common belief on the system state. This common belief can serve as an information state for a dynamic program that characterizes optimal control strategies. Even though our conceptual approach is borrowed from \cite{nayyar2013decentralized}, we have to  deal with  the infinite-dimensional strategy spaces of our problem and we cannot fully rely on the arguments in \cite{nayyar2013decentralized} that explicitly only deal with  finite strategy spaces.

%% file: system_model.tex
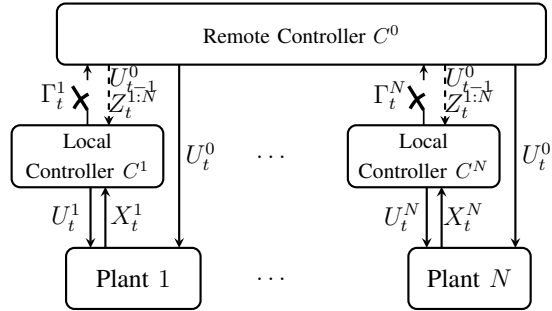
\begin{figure}
\begin{center}
\resizebox{0.45\textwidth}{!}{%
\begin{tikzpicture}[every text node part/.style={align=center}]
\begin{scope}[rotate=180]
\node [rectangle,draw,minimum width=8cm,minimum height=1cm,line width=1pt,rounded corners]at (2,-2) (1) {{\small Remote Controller $C^0$}}; 

\begin{scope}[shift={(0.25,0)}]
\node [rectangle,draw,minimum width=2.2cm,minimum height=1cm,line width=1pt,rounded corners,text width=2.2cm]at (-0.25,0) (2) {{\small Local Controller $C^{N}$}}; 
\node [rectangle,draw,minimum width=2cm,minimum height=1cm,line width=1pt,rounded corners]at (-1,2) (3) {Plant $N$}; 

\path[thick,->,>=stealth,line width=1pt]
           (-0.3,0.5) edge node {}   (-0.3,1.5) %$U_t^{L_1}$
           (-1.75,-1.5) edge node {}   (-1.75,1.5) %$U_t^R$
     ;
      
      \path[thick,->,>=stealth,line width=1pt]
           (-0.6,-1.5) edge node {}   (-0.6,-0.5) [dashed] %Sharing to C^L_{N}
      ;
      
    \path[thick,->,>=stealth, shift left=.30ex,line width=1pt]
        (-0.6,1.5) edge node {}   (-0.6,0.5); %$X^1_t$

\node[] at (.1,1) {$U_t^{N}$};
\node[] at (-0.9,1) {$X^N_t$};
\node[] at (-2.1,0) {$U_t^0$};
\node[] at (.3,-1) {$\Gamma_t^N$};
\node[] at (-1,-.85) {$Z_t^{1:N}$};
\node[] at (-1,-1.25) {$U_{t-1}^0$};
\end{scope}

\begin{scope}[shift={(-0.25,0)}]
\node [rectangle,draw,minimum width=2.2cm,minimum height=1cm,line width=1pt,rounded corners,text width=2.2cm]at (5.75,0) (2) {{\small Local Controller $C^{1}$}}; 
\node [rectangle,draw,minimum width=2.2cm,minimum height=1cm,line width=1pt,rounded corners]at (5,2) (4) {Plant $1$}; 

\path[thick,->,>=stealth,line width=1pt]
           (5.7,0.5) edge node {}   (5.7,1.5)           %$U_t^{L_N}$
           (4.25,-1.5) edge node {}   (4.25,1.5) %$U_t^R$
      ;
      
      \path[thick,->,>=stealth,line width=1pt]
           (5.4,-1.5) edge node {}   (5.4,-.5)   [dashed]        %Sharing to C^L_{1}
      ;
      
    \path[thick,->,>=stealth, shift left=.30ex,line width=1pt]
        (5.4,1.5) edge node {}   (5.4,0.5) ; % $X^N_t$

\node[] at (6.1,1) {$U_t^{1}$};
\node[] at (5.1,1) {$X^1_t$};
\node[] at (3.9,0) {$U_t^0$};
\node[] at (6.3,-1) {$\Gamma_t^1$};
\node[] at (5.0,-.85) {$Z_t^{1:N}$};
\node[] at (5.0,-1.25) {$U_{t-1}^0$};
\end{scope}

\node[] at (2.5,2) {$\ldots$};
\node[] at (2.5,0) {$\ldots$};
\end{scope}
\draw [line width=0.8pt,addarrow] (-5.5,0.5) to[cspst] (-5.5,1.5); %$\Gamma_t^1$
\draw [line width=0.8pt,addarrow] (0,0.5) to[cspst] (0,1.5); %$\Gamma_t^1$
%\draw  [->,>=stealth,line width=1pt] (-0.3,-.5) to coordinate[pos=.20] (A) coordinate[pos=.70] (B) (-0.3,-1.5) ; 
%\fill[white] let \p1=($(A)-(B)$), \n2={.5*veclen(\x1,\y1)-0.50pt}
%   in ($(A)!.5!(B)$) circle(\n2);
%\draw [line width=.5pt] (A) to[cspst] (B);
%
%\draw  [->,>=stealth,line width=1pt] (5.7,-.5) to coordinate[pos=.20] (A) coordinate[pos=.70] (B) (5.7,-1.5) ; 
%\fill[white] let \p1=($(A)-(B)$), \n2={.5*veclen(\x1,\y1)-0.50pt}
%   in ($(A)!.5!(B)$) circle(\n2);
%\draw [line width=.5pt] (A) to[cspst] (B);
\end{tikzpicture}
}
\caption{System model. The binary random variables $\Gamma_t^{1:N}$ indicate whether packets are transmitted successfully.}
\label{fig:SystemModel}
\end{center}
\end{figure}

%% file: timeline.tex
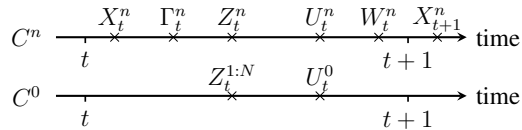
\begin{figure}[h]
\begin{center}
\resizebox{0.43\textwidth}{!}{%
\begin{tikzpicture}[node distance=1cm]
\begin{scope}[scale=0.99, every node/.append style={transform shape}]
\draw [thick,->,>=stealth,line width=1pt] (1,1) -- (8,1) node[right]{time};
\draw [thick] (1.5,1) -- (1.5,0.9) node[below] {$t$};
\draw [thick] (7,1) -- (7,0.9) node[below] {$t+1$};

\draw [thick,->,>=stealth,line width=1pt] (1,0) -- (8,0) node[right]{time};
\draw [thick] (1.5,0) -- (1.5,-0.1) node[below] {$t$};
\draw [thick] (7,0) -- (7,-0.1) node[below] {$t+1$};

%\node (Nat) at (0,2) {Nature};

\node (Xt) at (2,1.3) {$X_t^n$};
\draw (2,1) node[cross=2.5pt] {};

\node (Gamt) at (3,1.3) {$\Gamma_t^n$};
\draw (3,1) node[cross=2.5pt] {};

%\node (Gamt) at (2.5,0.3) {$\Gamma_t^{1:N}$};
%\draw (2.5,0) node[cross=2.5pt] {};

\node (Zt) at (4,1.3) {$Z_t^n$};
\draw (4,1) node[cross=2.5pt] {};

\node (Zt) at (4,0.3) {$Z_t^{1:N}$};
\draw (4,0) node[cross=2.5pt] {};

%\node (HtR) at (4,1.3) {$H_t^0$};
%\draw (4,1) node[cross=2.5pt] {};
%
%\node (HtR) at (4,0.3) {$H_t^0$};
%\draw (4,0) node[cross=2.5pt] {};
%
%\node (HtL) at (4.7,1.3) {$H_t^{n}$};
%\draw (4.7,1) node[cross=2.5pt] {};

\node (UtL) at (5.5,1.3) {$U_t^{n}$};
\draw (5.5,1) node[cross=2.5pt] {};

\node (UtL) at (5.5,0.3) {$U_t^{0}$};
\draw (5.5,0) node[cross=2.5pt] {};

\node (Wt) at (6.5,1.3) {$W_t^{n}$};
\draw (6.5,1) node[cross=2.5pt] {};

\node (Xtp) at (7.5,1.3) {$X_{t+1}^n$};
\draw (7.5,1) node[cross=2.5pt] {};

\node (CL) at (.5,1) {$C^{n}$} ;
\node (CR) at (.5,0) {$C^0$} ;

%\node (C) at (1.5,-1) {$C$} ;

\end{scope}
\end{tikzpicture}
}
\caption{Time-ordering of relevant variables.}
\label{fig:timing}

\end{center}
\end{figure}

%% file: DynamicProgram.tex
We first provide a structural result for the local controllers' strategies.
\begin{lemma}
\label{lm:first_structure}
For $n \in \mathcal{N}$, let $\hat H^{n}_t = \{X_{t}^n, H^0_t\}$, and $\hat{\mathcal{G}}^{n}=\{g^{n} \in \mathcal{G}^{n}: g^{n}_t \text{ depends only on }\hat H^{n}_t\}$.
Then,
\begin{align}
\inf_{g^{0} \in \mathcal{G}^{0},\ldots,g^{N} \in \mathcal{G}^{N}  } J(g^{0:N}) =\inf_{\substack{g^0\in
\mathcal{G}^0
\\ g^{1}\in
\mbox{\normalsize{$\hat{\mathcal{G}}^{1}$}}, \ldots, g^{N}\in 
\mbox{\normalsize{$\hat{\mathcal{G}}^{N}$}}
}} J(g^{0:N}).
\end{align}
\vspace{-2mm}
\end{lemma}
\begin{proof}
See Appendix \ref{Proof_Lm_first_structure}.
\end{proof}

Due to  Lemma \ref{lm:first_structure}, we only need to consider strategies $g^{n} \in \hat{\mathcal{G}}^{n}$ for the local controller $C^{n}$, $n \in \mathcal{N}$. That is, the local controller $C^{n}$ only needs to use $\hat H^{n}_t = \{X_{t}^n, H^0_t\}$ to make the decision at $t$.

According to the information structure \eqref{Model:info} and Lemma \ref{lm:first_structure}, $H_t^0$ is the common information among $C^{0:N}$, and $X_t^n$ (which is $\hat H_t^{n} \setminus H_t^0$) is the \emph{private} information used by the local controller $C^{n}$ in its decision-making.
Note that $C^0$ has no private information (since $H_t^0 \setminus H_t^0 = \emptyset$).
Following the common information approach \cite{nayyar2013decentralized}, we construct below an equivalent centralized problem using the controllers' common information.
\vspace{-2mm}
\subsection{Equivalent Centralized Problem}

Consider arbitrary control strategies $g^{n} \in \mathcal{\hat G}^{n}$, $n \in \mathcal{N},$ and $g^0 \in \mathcal{G}^0$ for the local and the remote controllers, respectively. Under these strategies, $U_t^n$ can be written as
\begin{align}
U_t^{n} &= g_t^{n}(X_t^n, H_t^0)
= \ee^{g} [g_t^{n}(X_t^n, H_t^0) \vert H_t^0] 
\notag \\
&+ \Big\{ g_t^{n}(X_t^n, H_t^0) -  \ee^{g} [g_t^{n}(X_t^n, H_t^0) \vert H_t^0] \Big\}.
\label{decomposition}
\end{align}
We can rewrite \eqref{decomposition} as 
\begin{align}
U_t^{n} = \bar g_t^{n}(H_t^0) + \tilde{g}_t^{n}(X_t^n, H_t^0),
\label{decomposition_abbr}
\end{align}
where 
\begin{align}
&\bar g_t^{n}(H_t^0) = \ee^{g} [g_t^{n}(X_t^n, H_t^0) \vert H_t^0], \notag \\
&\tilde{g}_t^{n}(X_t^n, H_t^0) = g_t^{n}(X_t^n, H_t^0) -  \ee^{g} [g_t^{n}(X_t^n, H_t^0) \vert H_t^0].
\end{align}
Observe that $\tilde{g}_t^{n}(X_t^n, H_t^0) $ is conditionally zero-mean given $H_t^0$, that is, $\ee^g[\tilde{g}_t^{n}(X_t^n, H_t^0) \vert H_t^0 ] =0$.

Note that $\bar g_t^{n}(H_t^0)$ is the conditional mean of $g_t^{n}(X_t^n, H_t^0)$ given the remote controller's information $H_t^0$ and $\tilde{g}_t^{n}(X_t^n, H_t^0)$ can be interpreted as the deviation of $g_t^{n}(X_t^n, H_t^0)$ from the mean $\bar g_t^{n}(H_t^0)$. With this interpretation, \eqref{decomposition_abbr} suggests that, at each time $t$, the problem of finding optimal control action $U_t^{n}$ for $C^{n}$ is equivalent to the problem of finding ``mean value" of $U_t^{n}$ and ``deviation" of $U_t^{n}$ from the mean value.

We will use the above representation of $g_t^{n}$ in terms of $\bar g_t^{n}$ and $\tilde{g}_t^{n}$ to formulate a centralized decision-making problem. In the centralized problem, the remote controller is the only decision-maker. At each time $t$, given the realization $h_t^0$ of the remote controller's information, it makes three decisions:
\begin{enumerate}
\item Remote controller's control action $u_t^0 = \phi_t^0 (h_t^0)$,
\item Mean value of every local controller's control action $\bar u_t^{n} = \bar \phi_t^{n}(h_t^0)$, $n \in \mathcal{N}$,
\item A ``deviation from the mean value" mapping $q_t^n \in \mathcal{Q}^n$, $n \in \mathcal{N}$, where \\ 
$\mathcal{Q}^n =\{q^n:\R^{d_X^n} \to \R^{d_U^n} \text{, Borel measurable} \}$\footnote{In other words, $\mathcal{Q}^n$ is the set of all Borel measurable functions from $\R^{d_X^n}$ to $\R^{d_U^n}$. }\\ and $q_t^n = \tilde \phi_t^{n}(h_t^0)$.
\end{enumerate}

Based on the above decisions, the control actions applied to the system described by \eqref{Model:system}-\eqref{Model:channel} are:
\begin{itemize}
\item  $u_t^0$, as the control action of the remote controller,
\item $U_t^{n} = \bar u_t^{n} + q_t^n(X_t^n)$, as the control action of the $n$-th local controller, $n \in \mathcal{N}$.
\end{itemize}

We call $u^{prs}_t:=(u^0_t,\bar u^{1:N}_t,q_t^{1:N})$ the prescription at time $t$.
We denote $(\phi^0_{t}, \bar \phi^{1:N}_{t}, \tilde \phi^{1:N}_{t})$ by $\phi^{prs}_{t}$ and write $u^{prs}_t = \phi^{prs}_{t}(h^0_t)$ 
to indicate that the prescription is a function of the common information $h^0_t$.
The functions $(\phi_t^{prs},  \hspace{2pt} 0 \leq t \leq T)$ are collectively referred to as the prescription strategy and denoted by $\phi^{prs}$.
The prescription strategy is required to satisfy the following conditions:
\begin{enumerate}[label=(\mylabel{C}{\arabic*})]
\item $\phi^0 \in \mathcal{G}^0$.
\item Define $\phi_t^{n}(X_t^n, H_t^0):= \bar \phi_t^{n}(H_t^0) + [\tilde \phi_t^{n}(H_t^0)](X_t^n) $. Then,
$\phi^{n} \in \mathcal{\hat G}^{n}$ for any $n\in\mathcal{N}$, where the notation $[\tilde \phi_t^{n}(H_t^0)](X_t^n)$ means that we first use $\tilde \phi_t^{n}(H^0_t)$ to find the deviation mapping $q^n_t$ and then evaluate $q^n_t$ at $X^n_t$.

\item 
We require that for any $t$, 
\begin{align}
\label{zero_mean_condition}
\ee^{\phi^{prs}} \Big\{[\tilde \phi_t^{n}(H_t^0)](X_t^n) \vert H_t^0\Big\} =0,
\end{align}
where $\ee^{\phi^{prs}}$ is the probability measure induced by the prescription strategy $\phi^{prs}$.

\end{enumerate}  

Denote by $\Phi^{prs}$ the set of all prescription strategies satisfying the above conditions. Consider the following problem of optimizing the prescription strategies.
\vspace{-2mm}
\begin{problem}
\label{problem:equivalent}
Consider the system described by \eqref{Model:system}-\eqref{matrix_structure}. Given a prescription strategy $\phi^{prs} \in \Phi^{prs}$, let 
\begin{align}
\Lambda  (\phi^{prs}) =& \ee^{\phi^{prs}} \left[\sum_{t=0}^T c_t^{prs}(X^{1:N}_t,U^{prs}_t)\right],
\label{Model:obj_equivalent}
\end{align}
where for any $x_t^{1:N}$ and $u^{prs}_t=(u^0_t,\bar u^{1:N}_t,q_t^{1:N})$,
\begin{align}
c_t^{prs}(x^{1:N}_t,u^{prs}_t) = c_t\Big(x_t^{1:N},u^{0}_t, \{\bar u^{n}_t+q^{n}_t(x_t^n)\}_{n\in\mathcal{N}} \Big).  
\end{align}

Then, we would like to solve the following optimization problem:

\vspace{-6mm}
\begin{align}
\inf_{\phi^{prs} \in \Phi^{prs}}  \Lambda  (\phi^{prs}).
\label{Model:optimization_equivalent}
\end{align}
\end{problem}

We now note that any feasible prescription strategy in Problem \ref{problem:equivalent} can be used to construct control strategies in Problem \ref{problem1}.
On the other hand, any control strategies in Problem \ref{problem1} can be represented by a prescription strategy in Problem \ref{problem:equivalent}. This equivalence between Problems \ref{problem1} and \ref{problem:equivalent} is formally stated in the following lemma.

\begin{lemma}
\label{thm:translation}
Problems \ref{problem1} and \ref{problem:equivalent} are equivalent in the following sense:
\begin{enumerate}
\item For any control strategies $g^{n} \in \hat{\mathcal{G}}^{n}$ and $g^0 \in \mathcal{G}^0$ in Problem \ref{problem1}, there is a prescription strategy $\phi^{prs}\in \Phi^{prs}$ in Problem \ref{problem:equivalent} such that for $0 \leq t \leq T$,

\vspace{-3mm}
\begin{small}
\begin{align}
& \phi^{0}_t(H_t^0) = g^{0}_t(H_t^0),
\label{eq:gR_equivalent1}
\\
&\bar \phi^{n}_t(H_t^0) = \bar g^{n}_t(H_t^0)
= \ee^{g} [g_t^{n}(X_t^n, H_t^0) \vert H_t^0], \hspace{1mm}\forall n \in \mathcal{N},
\label{eq:gLbar_equivalent1}
\\
&[\tilde \phi^{n}_t(H_t^0)](X_t^n)  = \tilde g^{n}_t(X_t^n,H_t^0)  \notag \\
&= g_t^{n}(X_t^n, H_t^0) -  \ee^{g} [g_t^{n}(X_t^n, H_t^0) \vert H_t^0]
, \hspace{1mm}\forall n \in \mathcal{N},
\label{eq:gL_equivalent1}
\\
&\Lambda (\phi^{prs}) = J(g^{0:N}).
\label{eq:obj_equivalent1}
\end{align}
\end{small}
\vspace{-3mm}
\item Conversely, for any prescription strategy $\phi^{prs} \in \Phi^{prs}$ in Problem \ref{problem:equivalent}, there are  
control strategies $g^{n} \in \hat{\mathcal{G}}^{n}$ and $g^0 \in \mathcal{G}^0$ in Problem \ref{problem1} such that for $0 \leq t \leq T$,

\vspace{-3mm}
\begin{small}
\begin{align}
& g^{0}_t(H_t^0) = \phi^{0}_t(H_t^0),
\label{eq:gR_equivalent2}
\\
& g^{n}_t(X_t^n, H_t^0)=  \bar g^{n}_t(H_t^0) + \tilde g^{n}_t(X_t^n,H_t^0) \notag \\
& \hspace{50pt} = \bar \phi^{n}_t(H_t^0) + [\tilde \phi^{n}_t(H_t^0)](X_t^n) , \hspace{1mm}\forall n \in \mathcal{N},
\label{eq:gL_equivalent2}
\\
& J(g^{0:N}) = \Lambda (\phi^{prs}).
\label{eq:obj_equivalent2}
\end{align}
\end{small}
\end{enumerate}
\end{lemma}
\vspace{-2mm}
\begin{proof}
The proof is a straightforward extension of arguments used in \cite{nayyar2013decentralized}, and is therefore omitted.
\end{proof}

\vspace{-4mm}
\subsection{Information State for Problem \ref{problem:equivalent}}

Since Problem \ref{problem:equivalent} is a centralized decision-making problem for the remote controller $C^0$, $C^0$'s belief on the system states can be used as an information state for decision-making. Note that $C^0$'s information at any time $t$ is the \emph{common information} $H^0_t$.
Therefore, we define the common belief $\Theta_t$ as the conditional probability distribution of $X_t^{1:N}$ given $H^0_t$. That is, under prescription strategies $\phi^{prs}_{0:t-1} $ until time $t-1$, for any measurable set $E \subset \prod_{n=1}^N \R^{d_X^n}$,
\begin{align}
\Theta_t (E) := 
\prob^{\phi^{prs}_{0:t-1}}(\vecc(X_t^{1:N}) \in E | H^0_{t}).
\label{eq:thetat}
\end{align}
Let $\Theta_t^n$ denote the marginal common belief on $X_t^n$. That is, for any measurable set $E^n \subset \R^{d_X^n}$
\begin{align}
\Theta_t^n(E^n) := 
\prob^{\phi^{prs}_{0:t-1}}(X_t^n \in E^n | H^0_{t}).
\label{eq:thetat_margin}
\end{align}
Then, for a given realization $h^0_t$ of $H^0_t$, the corresponding realization $\theta_t$ of $\Theta_t$ 
belongs to $\Delta(\prod_{n=1}^N \R^{d_X^n})$ and the realization $\theta^n_t$ of 
$\Theta^n_t$ belongs to $\Delta(\R^{d_X^n}), n \in \mathcal{N}$. 

Since the plants' dynamics are only coupled through the remote controller's actions which belongs to the common information, the common belief has the following conditional independence property.
\begin{lemma}
\label{lm:independence_of_states}
Consider a feasible prescription strategy $\phi^{prs} \in \Phi^{prs}$. Then, the random vectors $X_t^{1:N}$ are conditionally independent given the common information $H^0_t$. 
That is, for any measurable sets $E^n \subset \R^{d_X^n}$, $n \in \mathcal{N}$,
\begin{align}
&\Theta_t(\prod_{n=1}^N E^n)
=\prod_{n=1}^N  \Theta_t^n(E^n)
\label{eq:indep}
\end{align}
where $\Theta_t$ and $\Theta_t^n$ are given by \eqref{eq:thetat} and \eqref{eq:thetat_margin}.
\end{lemma}
\begin{proof}
The proof is a direct consequence of Part 2 of Claim \ref{lm:independe} in Appendix \ref{pre_results}. 
\end{proof}

\vspace{-6mm}
\begin{remark}
The conditional independence of the states $X_t^{1:N}$ given the common information, as described above in Lemma \ref{lm:independence_of_states}, is similar the conditional independence of the global state given the common information in \cite[Lemma 6]{mahajan2013optimal}. However, our model is different from the one considered in \cite{mahajan2013optimal}.
\end{remark}

From Lemma \ref{lm:independence_of_states}, the joint common belief $\Theta_t$ can be represented by the collection of marginal common beliefs $\Theta_t^{1:N}$.

We show in the following that the marginal common beliefs $\theta_t^n, n\in\mathcal{N},$ can be sequentially updated. 

\begin{lemma}
\label{lm:beliefupdate}
For any feasible prescription strategy $\phi^{prs} \in \Phi^{prs}$ and for any $h_t^0 \in \mathcal{H}_t^0$, we recursively define $\nu_t^n (h_t^0) \in \Delta(\R^{d_X^n})$ as follows:

For any measurable set $E^n \subset \R^{d^n_X}$,
\begin{align}
[\nu_0^n(h_0^0)] (E^n)
&= \left\{\begin{array}{ll}
\pi_{X_0^n}(E^n)
& \text{ if }z^n_{0}=\emptyset, 
\\
\mathds{1}_{E^n}(x^n_0) & \text{ if }z^n_{0}=x^n_{0}.
\end{array}\right.
\label{eq:theta0}
\\
[\nu_{t+1}^n(h_{t+1}^0)] (E^n)
&= [\psi_t^n(\nu_t^n(h_t^0),u^{prs}_t,z_{t+1}^n)](E^n) ,
\label{eq:theta_update}
\end{align}
where $u^{prs}_t= \phi^{prs}_{t}(h^0_t)$ and $\psi_t^n(\nu_t^n(h_t^0),u^{prs}_t,z_{t+1}^n)$ is defined as follows:
\begin{itemize}
\item If $z_{t+1}^n = x_{t+1}^n$, then 
\begin{align}
[\psi_t^n(\nu_t^n(h_t^0),u^{prs}_t,z_{t+1}^n)](E^n)= \mathds{1}_{E^n}(x^n_{t+1}).
\end{align}
\item If $z_{t+1}^n = \emptyset$, then 
\begin{small}
\begin{align}
&[\psi_t^n(\nu_t^n(h_t^0),u^{prs}_t, \emptyset)] (E^n)=
\nonumber
\\
&\int \int  \mathds{1}_{E^n}\big(f^n_t(x_{t}^n, w_{t}^n,u^{prs}_t)\big) 
\nu_t^n(h_t^0)(dx_t^n) \pi_{W_t^n}(dw_t^n),
\label{eq:psit}
\end{align}
\end{small}
where
\begin{align}
&f^n_t(x_{t}^n, w_{t}^n,u^{prs}_t) \notag \\
&= A^{nn} x_{t}^{n} + B^{nn}(\bar u^{n}_t+q_t^n(x_t^n))+B^{n0} u^0_t +w_t^n.
\label{eq:dynamic_function}
\end{align}
\end{itemize}
Then, $\nu_t^n$ is  the  conditional probability distribution of $X^n_t$ given $H^0_t$, that is $[\nu_{t}^n(H_{t}^0)] (E^n) = \prob^{\phi^{prs}_{0:t-1}}(X_t^n \in E^n | H^0_{t})$.
\end{lemma}
\begin{proof}
See Appendix \ref{Proof_Lm_belief_update} for a proof.
\end{proof}
Lemma \ref{lm:beliefupdate} implies that the realization $\theta_t^n$ of the belief $\Theta_t^n$ can be updated according to 
\begin{align}
\theta_{t+1}^n = \psi_t^n (\theta_t^n, u_t^{prs}, z_{t+1}^n).
\label{update_theta}
\end{align}

Recall that $\mathcal{Q}^n$ is the space of all measurable functions $q:\R^{d_X^n} \to \R^{d_U^n}$.
We now define the space $\mathcal{Q}^n(\theta^n) \subset \mathcal{Q}^n$ for any $\theta^n \in \Delta(\R^{d_X^n})$ to be

\vspace{-3mm}
\begin{small}
\begin{align}
&\mathcal{Q}^n(\theta^n) 
=
\Big\{q^n:\R^{d_X^n}\!\! \to \R^{d_U^n}\text{ measurable},\int q^n(x^n) \theta^n(d x^n) = 0
\Big\}.
\end{align}
\end{small}
Note that for any feasible prescription strategy $\phi^{prs}  \in \Phi^{prs}$, \eqref{zero_mean_condition} implies that
for almost every realization $h_t^0$ under $\phi^{prs}$,
\begin{align}
\ee^{\phi^{prs}}[q_t^n(X_t^n) \vert h_t^0] =0,
\label{zero_mean_condition_q}
\end{align}
where $q_t^n = \tilde \phi^{n}_t(h^0_t)$. Then, \eqref{zero_mean_condition_q} and \eqref{eq:thetat_margin} imply  that for almost every realization $h_t^0$,
$\int q_t^n(x_t^n) \theta_t^n (dx_t^n) =0$, that is, $q_t^n$ belongs to $\mathcal{Q}^n(\theta_t^n) $.

\subsection{Dynamic Program for Problem \ref{problem:equivalent}}

We can use the collection of marginal common beliefs $\Theta_t^{1:N}$ as an information state to construct a dynamic program for Problem \ref{problem:equivalent}. For that purpose, we will use the following definitions. 

For every $x \in \R^{d_X}$, we use $\rho(x)$ to denote the Dirac-delta distribution at $x$. Then,
for any $E \subset \R^{d_X}$, $[\rho (x)](E) = \mathds{1}_{E}(x)$.

For all $n \in \mathcal{N}$  and any $\theta_t^n \in \Delta(\R^{d_X^n})$, $q_t^n \in \mathcal{Q}^n(\theta^n_t)$, $\bar  u_t^{n} \in \R^{d_U^n}$,   $u_t^0 \in \R^{d_U^0}$,  and  $u^{prs}_t=(u^0_t,\bar u^{1:N}_t,q_t^{1:N})$, we define
\begin{itemize}
\item $\textit{IC}(\theta_t^{1:N}, u_t^{prs}) := 
\int  c_t^{prs}(x_t^{1:N}, u_t^{prs}) \prod_{n\in\mathcal{N}}\theta_t^n(dx_t^n)
$. This function represents the remote controller's expected instantaneous cost at time $t$ when its beliefs on the system states are $\theta_t^{1:N}$ and it selects $u_t^{prs}$.

\item $\alpha_t^n := \psi_t^n(\theta_t^n,u^{prs}_t, \emptyset)$ (see \eqref{update_theta} and note that $\alpha_t^n \in \Delta(\R^{d_X^n})$).
\item For any realization $\gamma_{t+1}^n \in \{0,1\}$ of $\Gamma_{t+1}^n$, $\textit{NB}(\gamma_{t+1}^n, \alpha_t^n, x_{t+1}^n) := (1- \gamma_{t+1}^n) \alpha_t^n + \gamma_{t+1}^n \rho(x_{t+1}^n)$. 
This function represents the next belief equation for $\theta_t^n$. If $\gamma_{t+1}^n = 0$, $\theta_{t+1}^n = \alpha_t^n$ and if $\gamma_{t+1}^n = 1$, $\theta_{t+1}^n = \rho(x_{t+1}^n)$.

\item $\textit{LS}(p^n, \gamma_{t+1}^n) := (p^n)^{1-\gamma_{t+1}^n}(1-p^n)^{\gamma_{t+1}^n}$. 
If $\gamma_{t+1}^n = 0$, this function represents the link failure probability, that is $p^n$. 
If $\gamma_{t+1}^n = 1$, this function represents the probability that link is active, that is $1- p^n$.

\end{itemize}

The following theorem provides a dynamic program for optimal prescription strategies of Problem \ref{problem:equivalent}.

\begin{theorem}
\label{thm:structure}
Suppose there exist functions $ V_{t}:\prod_{m=1}^N \Delta(\R^{d_X^m}) \to \R \text{ for }t=0,1,\dots,T+1$ such that for each $\theta_t^{1:N}\in \prod_{m=1}^N \Delta(\R^{d_X^m})$, the following are true:
\begin{itemize}
\item $V_{T+1}(\theta_t^{1:N}) = 0$,
\item For any $t = 0,1,\dots,T$
\begin{small}
\begin{align}
&V_t(\theta_t^{1:N}) =  \min_{\{q_t^n\in\mathcal{Q}^{n}(\theta^n_t)\}_{n \in \mathcal{N}} } \Big\lbrace \min_{\{\bar u^{n}_t\in\R^{d_U^n}\}_{n \in \mathcal{N}}, u^0_t\in\R^{d_U^0}} \Big \lbrace
\notag\\
& \textit{IC}(\theta_t^{1:N}, u_t^{prs}) 
+  \sum_{\gamma_{t+1}^1\in\{0,1\}} \ldots \sum_{\gamma_{t+1}^N \in\{0,1\}}  \prod_{n \in \mathcal{N}} \textit{LS}(p^n, \gamma_{t+1}^n)
\notag\\
&\times 
\int  V_{t+1} \Big( \Big\{ 
\textit{NB}(\gamma_{t+1}^n, \alpha_t^n, x_{t+1}^n)
  \Big\}_{n \in \mathcal{N}} \Big)
\prod_{n\in\mathcal{N}}\alpha_t^n(dx_{t+1}^n)  \Big \rbrace \Big \rbrace,
\label{eq:DP_V}
\end{align}
\end{small}
where $u^{prs}_t=(u^{0}_t,\bar u^{1:N}_t,q_t^{1:N})$.
\end{itemize}
Further, suppose there exists a feasible prescription strategy
$\phi^{prs*}\in \Phi^{prs}$
such that for any realization $h^0_t\in\mathcal{H}^0_t$ and its corresponding common beliefs $\theta_t^n=\nu_t^{n}(h_t^0)$, $n \in \mathcal{N}$, (as defined by Lemma \ref{lm:beliefupdate}), the prescription $u^{prs*}_t=(u^{0*}_t,\bar u^{1:N*}_t,q_t^{1:N*})=\phi^{prs*}(h^0_t)$
achieves the minimum in the definition of $V_{t}\big(\theta^{1:N}_t\big)$. Then, $\phi^{prs*}$ is an optimal prescription strategy for Problem \ref{problem:equivalent}.
\end{theorem}
\begin{proof}
See Appendix \ref{Proof_Thm_structure} for a proof.
\end{proof}

If the functions $V_{0:T}$ of Theorem \ref{thm:structure} can be shown to exist, then Theorem \ref{thm:structure} provides a dynamic program to solve Problem \ref{problem:equivalent}. Even if such a dynamic program is available, it suffers from two significant challenges. First, it is a dynamic program on the belief space $\prod_{n=1}^N\Delta(\R^{d_X^n})$ which is infinite dimensional. 
Second, each step of the dynamic program involves  functional optimization over the spaces $\mathcal{Q}^n(\theta^n_t)$, $n \in \mathcal{N}$. 
In the next section, we show that  functions satisfying \eqref{eq:DP_V} exist and that
it is possible to use the dynamic program of  Theorem \ref{thm:structure} to obtain optimal control strategies in Problem \ref{problem1}.

%% file: OptimalStrategies.tex
\subsection{Optimal prescription strategy in Problem \ref{problem:equivalent}}
For a vector $x$ and a matrix $G$, we use
\begin{align}
QF(G,x):= x^\tp G\, x = \tr(Gxx^\tp)
\label{eq:QF}
\end{align}
to denote the quadratic form.
We define the operators $\Omega$ and $\Psi$ as follows,

\vspace{-3mm}
\begin{small}
\begin{align}
&\Omega (P,A,B,R^{11},R^{22}, R^{12}) = R^{11} + A^{\tp} P A \notag \\
& - (R^{12} + A^{\tp} P B) 
(R^{22} + B^{\tp} P B)^{-1} \big((R^{12})^{\tp} + B^{\tp} P A\big), \\
&\Psi (P,A,B,R^{22}, R^{12}) = - (R^{22} + B^{\tp} P B)^{-1} \big((R^{12})^{\tp} + B^{\tp} P A \big).
\end{align}
\end{small}
\vspace{-3mm}
\begin{remark}
$P = \Omega (P,A,B,R^{11},R^{22}, R^{12})$ is the discrete-time algebraic Riccati equation.
\end{remark}

The following theorem presents  functions $V_{0:T}$ satisfying \eqref{eq:DP_V} and an explicit optimal solution of the dynamic program in Theorem \ref{thm:structure}.
\begin{theorem}
\label{thm:Sol_packetdrop}

For $t=0,1,\ldots,T$, the functions
$V_t(\cdot)$ of Theorem \ref{thm:structure} exist and are given by\footnote{Recall that $\mu(\theta_t^n)$ and 
$\cov(\theta_t^n)$ are the mean vector and the covariance matrix for the probability distribution $\theta_t^n$.}  

\vspace{-2mm}
\begin{small}
\begin{align}
V_t(\theta_t^{1:N}) &= 
QF \Big( P_t, \vecc \big(\{\mu(\theta_t^n) \}_{n \in \mathcal{N}} \big) \Big)
 \notag \\& 
 + \sum_{n=1}^N \tr \Big( \tilde{P}_t^{nn} \cov(\theta_t^n) \Big) 
  + e_t,
\label{eq:Vt_PacketDrop}
\end{align}
\end{small}
\vspace{-2mm}
where 
\begin{small}
\begin{align}
e_{t} = \sum_{s = t}^T \sum_{n=1}^N \tr \Big( ((1-p^n)P_{s+1}^{nn}+p^n\tilde P_{s+1}^{nn})\cov(\pi_{W_s^n}) \Big).
\label{eq:error}
\end{align}
\end{small}
The matrices $P_t$,$\tilde P_t^{nn}$, $n \in \mathcal{N}$, defined recursively below, are symmetric positive semi-definite (PSD) and $P^{nn}_t$ is the $n$th diagonal block of $P_t$.

\vspace{-4mm}
\begin{small}
\begin{align}
\label{recursion_P_last}
P_{T+1} &=0 \\ 
\label{recursion_P}
P_{t} &= \Omega (P_{t+1},A,B,R_t^{XX},R_t^{UU}, R_t^{XU}), \\
K_t &= \Psi (P_{t+1},A,B,R_t^{UU}, R_t^{XU});
\label{Gain_K}
\end{align}
\end{small}
For $n \in \mathcal{N}$,
\begin{small}
\begin{align} 
\label{recursion_tilde_P_last}
\tilde P_{T+1}^{nn} &= \mathbf{0}, \\
\tilde P_{t}^{nn} &= \notag \\
&\hspace{-7mm} \Omega \Big((1-p^n) P_{t+1}^{nn} + p^n \tilde P_{t+1}^{nn},A^{nn},B^{nn},R_t^{X^nX^n},R_t^{U^nU^n}, R_t^{X^nU^n} \Big), 
\label{recursion_tilde_P}
\\
\tilde K_t^{nn} &= \Psi \Big((1-p^n) P_{t+1}^{nn} + p^n \tilde P_{t+1}^{nn},A^{nn},B^{nn},R_t^{U^nU^n}, R_t^{X^nU^n} \Big).
\label{Gain_tilde_K}
\end{align}
\end{small}
Furthermore, the optimal prescription strategy is given as,
\begin{align}
\begin{bmatrix}
u^{0*}_t \\
\bar u^{1*}_t \\
\vdots \\
\bar u^{N*}_t 
\end{bmatrix} 
&=
\begin{bmatrix}
\phi^{0*}_t (\theta_t^{1:N}) \\
\bar \phi^{1*}_t (\theta_t^{1:N})\\
\vdots \\
\bar \phi^{N*}_t (\theta_t^{1:N})
\end{bmatrix} 
=K_t
\begin{bmatrix}
\mu(\theta_t^1) \\
\vdots \\
\mu(\theta_t^N)
\end{bmatrix},
\label{eq:opt_ubar}
\\
q^{n*}_t(x_t^n) &= [\tilde \phi^{n*}_t(\theta_t^{1:N})](x_t^n) =\tilde K_t^{nn} \Big( x_t^n - \mu (\theta_t^n) \Big).
\label{eq:opt_gamma}
\end{align}

\end{theorem}

\vspace{-3mm}
\begin{proof}
See Appendix \ref{Proof_Thm_optimal_solution} for a proof.
\end{proof}
\vspace{-2mm}
\begin{remark}
\label{rem:noGaussian}
\mbox{}
\begin{enumerate}
\item Consider a centralized version of our problem where at each time $t$, the remote controller knows all the states $X_t^{1:N}$ and chooses all control actions $U_t^{0:N}$. The solution to this centralized problem is $U_t = K_t X_t$ where $K_t$ is as given in \eqref{Gain_K} and \eqref{recursion_P_last}-\eqref{recursion_P} are the standard Riccati recursion for the centralized problem.
\item The recursions of Theorem \ref{thm:Sol_packetdrop}, and hence the optimal prescription strategy, do not depend on the covariances of the initial state and the noises.
\item Note that the dimension of matrices $P_t$ in \eqref{recursion_P_last}-\eqref{recursion_P} increases with the number $N$ of local controllers. Hence, the complexity of doing the recursions in \eqref{recursion_P_last}-\eqref{recursion_P} increases with $N$.
\end{enumerate}
\end{remark}

\subsection{Optimal control strategies in Problem \ref{problem1}}
From Theorem \ref{thm:structure}, Theorem \ref{thm:Sol_packetdrop} and Lemma \ref{thm:translation}, we can  explicitly compute the optimal control strategies for Problem~\ref{problem1}.

\begin{theorem}
\label{thm:opt_strategies}
The optimal strategies of Problem \ref{problem1} are given by
\begin{align}
\begin{bmatrix}
U^{0*}_t \\
\bar U^{1*}_t \\
\vdots \\
\bar U^{N*}_t 
\end{bmatrix} &=
K_t
\begin{bmatrix}
\hat X_t^1 \\
\vdots \\
\hat X_t^N
\end{bmatrix},
\label{eq:opt_ULbarUR}
\\
U^{n*}_t
&= \bar U^{n*}_t + \tilde K_t^{nn} \left(X_t^n - \hat X_t^n\right), \hspace{1mm} n \in \mathcal{N}
\label{eq:opt_UL}
\end{align}
where $\hat X_t^n$ is the estimate (conditional expectation) of $X_t^n$ based on the common information $H^0_t$. $\hat X_t^n$ can be computed recursively according to
\begin{align}
&\hat X^n_0 = \left\{
\begin{array}{ll}
\mu(\pi_{X_0^n}) & \text{ if }Z_{0}^n= \emptyset,\\
 X_{0}^n & \text{ if }Z_{0}^n = X_{0}^n.
\end{array}\right.
\label{eq:estimator_0}
\end{align}
\begin{align}
&\hat X_{t+1}^n 
= \left\{
\begin{array}{ll}
 A^{nn} \hat X_t^n + B^{nn}\bar U^{n*}_t + B^{n0} U^{0*}_t & \text{if }Z_{t+1}^n= \emptyset,\\
 X_{t+1}^n & \text{if }Z_{t+1}^n = X_{t+1}^n.
\end{array}\right.
\label{eq:estimator_t}
\end{align}
\end{theorem}
\begin{proof}
See Appendix \ref{Proof_Thm_optimal_strategies} for a proof.
\end{proof}

Theorem \ref{thm:opt_strategies} shows that the optimal control strategy of the remote controller $C^0$ is linear in the  state estimate $\hat X_t^{1:N}$, and the optimal control strategy of the local controller $C^{n}$, $n \in \mathcal{N}$, is linear in both the state $X_t^n$ and the  state estimate $\hat X_t^{1:N}$.

\begin{remark}
According to Theorems \ref{thm:Sol_packetdrop} and \ref{thm:opt_strategies}, gain matrices $K_{0:T}$, $\tilde K_{0:T}^{nn}$, $n \in \mathcal{N}$, are calculated offline and only the estimates $\hat X_t^{1:N}$ are computed online at each time $t$ using \eqref{eq:estimator_0}-\eqref{eq:estimator_t}.
\end{remark}

\vspace{-4mm}
\subsection{Simplification of communication from remote controller}
The assumption that the remote controller can perfectly send $Z_t^{1:N}$ and $U_{t-1}^{0*}$ to each local controller requires a lot of communication resources such as bandwidth, especially since the amount of   information  to send grows with the number $N$ of local controllers.
One way to simplify the communication is to transmit only the information needed for computing the optimal control actions.
According to Theorem \ref{thm:opt_strategies}, the local controller $C^{n}$ only needs to know $\bar U^{n*}_t$ and $\hat X_t^n$ (in addition to $X^n_t$ which it observes)
 to compute the optimal control action at time $t$. Hence, instead of sending $Z_t^{1:N}$ and $U_{t-1}^{0*}$ at time $t$, the remote controller can send only $\bar U^{n*}_t$ and $\hat X_t^n$ to  local controller $C^{n}$. By doing so, the communication resources required for each transmission are fixed and do not change with the number of controllers.
 
The above observation also shows that  resource-constrained channels  from $C^0$ to $C^{n}$ can achieve the same performance as perfect channels as long as the channel allows the vector $\vecc(\hat X_t^n,\bar U^{n*}_t) \in \R^{d_X^n+d^n_U}$ to be perfectly transmitted from $C^0$ to $C^{n}$.

Note that the communication can be further simplified when $d_U^0 < d^n_X$. In this case, the remote controller can send $\Gamma_{t}^n$ and $\vecc(U_{t-1}^{0*},\bar U^{n*}_t) \in \R^{d_U^0+d^n_L}$ to the local controller $C^{n}$. Then, the local controller $C^{n}$ can first compute $\hat X_t^n$ using \eqref{eq:estimator_t} and then  use $\hat X_t^n$ and $\bar U^{n*}_t$ to compute the optimal control action $U^{n*}_t$ based on \eqref{eq:opt_UL}.

\vspace{-3mm}
\subsection{Special cases}

\subsubsection{No control action for some controllers}
Our model can also capture the situation when some controllers participate in the communication but do not take any control action. In particular, the situation when controller $C^n$, $n \in  \overline{\mathcal{N}}$ has no action can be captured in the system model of Section \ref{sec:model} by setting $B^{nn} = \mathbf{0}$ (if $n \in \mathcal{N}$), or $B^{m0} = \mathbf{0}$ for all $m\in \mathcal{N}$ (if $n=0$), $R_t^{U^nU^n} = \mathbf{I}$, and $R_t^{XU^n},R_t^{U^nX},R_t^{U^mU^n},R_t^{U^nU^m}$ to be zero matrices for all $m \in \overline{\mathcal{N}} \setminus \{n\}$. Then, from Theorem \ref{thm:opt_strategies}, the optimal action $U^{n*}_t$ is zero which means that controller $n$ takes no action.

\subsubsection{Decoupled systems}
Consider the system model of Section \ref{sec:model} where the dynamics of plant $n$ in \eqref{Model:system}
and the instantaneous cost of sub-system $n$ (that is, plant $n$ and the local controller $C^n$ collectively)
in \eqref{cost_function} are affected only by the $n$-th component of the remote controller's action $U_t^0$, denoted by $[U^{0}_t]{\color{black} _{n \bullet} }$. Specifically,
\begin{small}
\begin{align}
&X_{t+1}^n = A^{nn} X_t^n + B^{nn}U^{n}_t+ \bar B^{n0} [U^{0}_t]{\color{black} _{n \bullet} } + W_t^n, t=0,\dots,T,
\notag
\\
&c_t(X_t^{1:N}, U^{0:N}_t) =
\sum_{n=1}^N c_t^n(X_t^n, [U^{0}_t]{\color{black} _{n \bullet} },U^{n}_t),
\end{align}
\end{small}
where $c_t^n$ is a quadratic function of the form \eqref{cost_function}.

We can still use Theorem \ref{thm:opt_strategies} to find optimal control strategies in this model. However, it is more efficient to consider the system as consisting of $N$ decomposed remote controllers $C^{0n}, n \in \mathcal{N}$, where the remote controller $C^{0n}$ is associated with only subsystem $n$. The problem of finding optimal strategies then decomposes into $N$ separate problems, each with one remote and one local controller. Each subproblem is a special case of Problem \ref{problem1}. Problems with one local and one remote controller were also investigated in our prior work  \cite{Ouyang_Asghari_Nayyar:CDC_2016}.

\subsubsection{Always active links}
Consider an instance of Problem \ref{problem1} where the links from the local controllers to the remote controller are always active, that is, $\Gamma_t^n =1$, for all $n \in \mathcal{N}$, and all $t = 0,\ldots,T$. Note that in this case, we have $Z_t^n = X_t^n$ for all $n \in \mathcal{N}$, $t = 0,\ldots,T$.  Hence, Problem \ref{problem1} effectively becomes a centralized problem. The optimal strategies of this problem can be calculated using Theorem \ref{thm:opt_strategies} as $U_t^* = K_t X_t$ where $K_t$ is computed recursively using \eqref{recursion_P_last}-\eqref{Gain_K}. These results are identical to the standard results for centralized LQ control problem under the cost function of \eqref{cost_function}.

\subsubsection{Always failed links}
Consider an instance of Problem \ref{problem1} where the links from the local controllers to the remote controller are always failed, that is, $\Gamma_t^n =0$, for  all $n \in \mathcal{N}$, $t = 0,\ldots,T$. In this case, the optimal control strategies  are given by,

\vspace{-3mm}
\begin{small}
\begin{align*}
\begin{bmatrix}
U^{0*}_t \\
 U^{1*}_t \\
\vdots \\
 U^{N*}_t
\end{bmatrix} &=
K_t
\begin{bmatrix}
\hat X_t^1 \\
\vdots \\
\hat X_t^N
\end{bmatrix} 
+ \begin{bmatrix}
   \mathbf{0} & & & \\
 &  \tilde K_t^{11} & & \mathlarger{\mathlarger{\mathlarger{\mathlarger{\mathlarger{0}}}}} \\
       &   & \ddots & \\
     & \mathlarger{\mathlarger{\mathlarger{\mathlarger{\mathlarger{0}}}}} & & \tilde K_t^{NN}
\end{bmatrix}
\begin{bmatrix}
\mathbf{0} \\
X_t^1 - \hat X_t^1 \\
\vdots \\
X_t^N - \hat X_t^N
\end{bmatrix},
\end{align*}
\end{small}
where $K_t$ is computed recursively using \eqref{recursion_P_last}-\eqref{Gain_K}. Furthermore, $\tilde K^{nn}_t$ is computed recursively  using \eqref{recursion_tilde_P_last}-\eqref{Gain_tilde_K} by setting $p^n = 1$ for all $n \in \mathcal{N}$. 

Note that in the case of always failed links, we have $Z_t^n = \emptyset$ for all $n \in \mathcal{N}$, $t = 0,\ldots,T$. According to \eqref{Model:info}, in this case 
$H^0_t = \{U^{0}_{0:t-1}\}$, $H^{n}_t= \{X^n_{0:t}, U^{n}_{0:t-1}\} \cup H^0_t$. 
 Problem 1 now becomes a $(N+1)$-controller problem with partially nested information structure. More specifically, in this case  $C^0$'s action $U_{t-1}^0$ affects the remote controller $C^n$'s information $H_t^n$, but since $H_{t-1}^0 \subset H_t^n$, the information structure is partially nested. This partially nested $(N+1)$-controller problem  has been studied in \cite{lessard_2012} with  linear time-invariant plants and controllers,  continuous-time dynamics  and with the objective of finding the optimal linear time-invariant controllers that minimize an infinite-horizon quadratic cost.

\begin{remark}
A  setup  with partially nested information structure similar to \cite{lessard_2012} but with finite state and action spaces was studied in \cite{Wu_Lall}. However, \cite{Wu_Lall} only provides a dynamic program without explicitly solving it. The finite-ness of state/action spaces, the absence of unreliable communication and the lack of an explicit solution make this work very different from ours.
\end{remark}

%% file: Discussions.tex
\subsection{Common Information Approach}
As stated before, our approach for Problem \ref{problem1} is conceptually based on the common information approach of \cite{nayyar2013decentralized}. Due to the perfect downlinks, the remote controller's information $H_t^0$ is common information among all controllers and hence it can serve as a \emph{coordinator} who provides \emph{prescriptions} to all controllers to compute their optimal control actions. Although we conceptually follow the common information approach of \cite{nayyar2013decentralized}, we had to come up with some new technical arguments to adapt this approach to our problem. The technical argument in \cite{nayyar2013decentralized} is proven for finite state and action spaces. While the authors in \cite{nayyar2013decentralized} state that their results should apply to more general spaces, this was not explicitly proven.  In our model, both the state and the action spaces are Euclidean. This has several implications: 
\balphlist
\item Firstly, unlike the case with finite state and action spaces where the set of feasible strategies is a finite set, the set of feasible strategies in our problem is an infinite-dimensional space. This is not merely a difference in the size of the problem. This difference means that in our version of the coordinator's problem, the common belief is a conditional probability measure on a Euclidean space and the coordinator's decision is to be selected from an infinite-dimensional space of all mappings from one Euclidean space to another. 

\item Because of the features of the coordinator's problem described above, it is not known a priori whether well-defined, measurable value functions satisfying the dynamic program of Theorem 1 actually do exist. Note that this existence was trivially true in the finite case of \cite{nayyar2013decentralized}.

\item Further, in the dynamic program of Theorem 1, it is not known a priori if a minimizing prescription for the coordinator exists at each step of the dynamic program and for each possible common belief. Even if such minimizing prescriptions were known to exist, it is still unclear whether a coordination strategy that selects the minimizing prescription for each possible common belief is even measurable. Clearly, if a minimizer-selecting strategy is not measurable, it is not feasible because we cannot even define the expectations involved in the problem. Due to these reasons, our Theorem 1 provides only sufficient conditions for optimality--- Theorem 1 is useful only if well-defined, measurable value functions and minimizer-selecting strategies can be shown to exist.  All of these difficulties are trivially absent from \cite{nayyar2013decentralized} due to the assumed finiteness of spaces involved.
\end{list}

While other works have used common information approach for linear strategies \cite{LessardNayyar:2013, mahajan2015sufficient}, these again bypass the technical difficulties described above. This is because (i) linear strategies imply a finite-dimensional strategy space, (ii) in the context of LQG problems, linear strategies result in Gaussian common beliefs which can be replaced by mean and covariance in the coordinator's dynamic program. Thus, both the belief space and the set of prescriptions are finite-dimensional Euclidean spaces and one can use straightforward modifications of \cite{nayyar2013decentralized} here.

 To the best of our knowledge, this is the first paper to explicitly show that the common information approach for decentralized control is not confined to the realm of problems where state/action spaces are finite or problems which pre-suppose linear strategies. Even though our strategy space allows for arbitrary measurable functions, we were able to adapt the common information approach to find explicit optimal strategies.

\vspace{-3mm}
\subsection{Structure of Optimal Controllers}
As discussed in Section \ref{sec:model}, the information structure of Problem \ref{problem1} is not partially nested due to the unreliable links. Nevertheless, the information structure of Problem \ref{problem1} behaves in a way similar to a partially nested structure in the following sense: If the states of uplinks  are fixed a priori to a certain realization, that is, $\Gamma^{1:N}_t = \gamma^{1:N}_t \in \{0,1\}^N$ for all $t$, the information structure of the problem becomes partially nested (see also the special cases 3) and 4) in Section \ref{sec:solution}-D). 
Therefore, we would expect that optimal controllers have a linear structure if the realization of $\Gamma^{1:N}_t, t=0,\ldots,T$,  was known in advance.
Theorem \ref{thm:opt_strategies} shows that the linearity of controllers is true even when $\Gamma^{1:N}_t, t=0,\ldots,T, $ are  not fixed in advance but only causally observed --- given the realization of $\Gamma^{1:N}_{0:t}$, the common estimates $\hat X^{1:N}_t$ are linear in the available common  information, and the optimal control action of $C^n$ is linear in the actual state $X^n_t$ and the common state estimates $\hat X^{1:N}_t$.

%% file: NumericalExperiments.tex
 In this section, we apply the result of Theorem \ref{thm:opt_strategies} to an instance of Problem \ref{problem1} and its corresponding centralized LQ problem (see the special case 3 in Section \ref{sec:solution}-D). 
The purpose of this example is to show that finding optimal strategies for an arbitrary instance of Problem \ref{problem1} using our results is computationally efficient and it is computationally comparable to a corresponding centralized LQ problem.

Consider an instance of Problem \ref{problem1} with one remote controller and $N$ local controllers over a time horizon of duration $T=1000$. We assume that $d_X^n = d_X=3$ for all $n \in \mathcal{N}$ and $d_U^n = d_U=3$ for  all $n \in \overline{\mathcal{N}}$. We want to measure the running time of computing the optimal control strategies for this problem and its corresponding centralized LQ problem.
In order to make sure that our comparison does not depend on the particular choices of system matrices, we calculate the running time for $100$ different instances of Problem \ref{problem1} and their corresponding centralized LQ problem. For each iteration,  each entry of the system matrices is chosen randomly and independently according to a uniform distribution on the interval  $[0,20]$. 
In particular, in each iteration, we generate a matrix $A$, a matrix $B$, and a collection of symmetric PD matrices $R_{0:T}$ randomly\footnote{To generate a $d \times d$ symmetric PD matrix, we  generate $\frac{d(d+1)}{2}$ numbers randomly. Then, we check to see whether the resulting symmetric matrix is PD. If not, we repeat this process until we generate a PD matrix.}. Further, the random variables $X_0^{1:N}$ and $W_{0:T}^{1:N}$ are chosen according to independent Gaussian distributions with zero mean and identity covariance matrices\footnote{According to Remark \ref{rem:noGaussian}, computation of optimal strategies does not depend on the covariances of initial states and noises.}. For this problem setup, we perform the following two experiments:

\begin{itemize}

\item We generate a set of random variables $\Gamma_t^n$, $n \in \mathcal{N}$, $t=0,\ldots,T,$ according to a Bernoulli distribution with $\prob(\Gamma_t^n = 0) = p^n =0.5$ for all $n \in \mathcal{N}$. We use Theorem  \ref{thm:opt_strategies} to compute the optimal control strategies and measure the time required for this computation. 

\item Next, we fix $\Gamma_t^n = 1$, for all $n \in \mathcal{N}$, $t=0,\ldots,T$. In this case, Problem 1 effectively becomes a centralized problem. We use special case 3 in Section \ref{sec:solution}-D to calculate the optimal control strategies and measure the time required for this computation. Note that the runtime for this experiment 
is simply the time required for computing the optimal control strategies for a \emph{centralized LQ problem} with the aforementioned system matrices.
\end{itemize}

The experiments were performed on a MacBook Pro, Intel 3 GHz core i7 processor with 16 GB memory. 
Tables \ref{table_example_1} and \ref{table_example_2} show the average running time of  instances of Problem \ref{problem1} with unreliable links and their corresponding centralized LQ problems (with always active links) for  different values of $N$. As can be seen from Tables \ref{table_example_1} and \ref{table_example_2}, applying our results to an arbitrary instance of Problem \ref{problem1} is computationally comparable to finding  optimal strategies in its corresponding centralized LQ problem. 
  
\begin{table}
\caption{Average running time in seconds for computing the optimal strategies for 100 instances of Problem \ref{problem1} with unreliable links}
\vspace{-7mm}
\label{table_example_1}
\begin{center}
\begin{small}
\begin{tabular}{ l | c | c | c | c}
  \hline			
  $N$ (\# of local controllers) & 1 & 10& 100 & 1000  \\  \hline
  Average running time (s)  & 0.347 & 1.191 & 32.353 & 19163.20 \\ 
  \hline  
\end{tabular}
\vspace{-6mm}
\end{small}
\end{center}
\end{table}
\vspace{-6mm}
\begin{table}
\caption{Average running time in seconds for computing the optimal strategies for the corresponding centralized LQ problems (with always active links)}
\vspace{-7mm}
\label{table_example_2}
\begin{center}
\begin{small}
\begin{tabular}{ l | c | c | c | c}
  \hline			
  $N$ (\# of local controllers) & 1 & 10& 100 & 1000  \\  \hline
  Average running time (s)  & 0.101 & 0.441 & 26.015 & 18512.49 \\ 
  \hline  
\end{tabular}
\vspace{-6mm}
\end{small}
\end{center}
\end{table}

%% file: Conclusion.tex
 We considered a networked control system (NCS) consisting of a remote controller and a collection of linear plants, each associated with a local controller. 
Each local controller directly observes the state of its co-located plant and can inform the remote controller of the plant's state through an unreliable uplink channel. The downlink channels from the remote controller to local controllers are assumed to be perfect. 
The objective of the local controllers and the remote controller is to cooperatively minimize a quadratic performance cost.
This multi-controller NCS problem is not a partially nested LQG problem, hence we cannot directly use prior results in decentralized control to conclude that linear strategies are optimal.

We employed the common information approach to this problem and showed that it is equivalent to a centralized sequential decision-making problem where the remote controller is the only decision-maker. We provided a dynamic program to obtain optimal strategies in the equivalent problem. Then, using these optimal strategies for the equivalent problem, we obtained optimal control strategies for all local controllers and the remote controller in our original problem. In the optimal control strategies, all controllers compute common estimates of the states of the plants based on the common information obtained from the communication network. The remote controller's action is linear in the common state estimates, and the action of each local controller is linear in both the actual state of its corresponding plant and the common state estimates.

Our results sketch a solution methodology for decentralized control with unreliable communication among controllers.
The methodology can potentially be generalized to other communication topologies in decentralized control such as directed acyclic communication graphs with unreliable links.

%% file: Appendices.tex
\section{Preliminary Results}
\label{pre_results}

\newtheorem{claim}{Claim}
In this section, we state and prove a set of claims which are useful in proving the main results of this paper.

\begin{claim}
\label{lm:conditional_independence}
Let $\mathscr{F}^0$, $\mathscr{F}^{1:N}$ and $\mathscr{G}^{1:N}$ be $\sigma-$algebras such that $\mathscr{F}^{1:N}$ are conditionally independent given $\mathscr{F}^0$, and $\mathscr{G}^n \subset \mathscr{F}^n$, $n \in \mathcal{N}$. Then, for $A^n \in \mathscr{F}^n$, $n \in \mathcal{N}$,\footnote{By conditioning on multiple $\sigma-$algebras, we mean conditioning on the smallest $\sigma-$algebra containing these $\sigma-$algebras.}

\vspace{-3mm}
\begin{small}
\begin{align}
\prob(\{A^n\}_{n \in \mathcal{N}} \vert \mathscr{F}^0, \mathscr{G}^{1:N})
= \prod_{n \in \mathcal{N}}  \prob(A^n \vert \mathscr{F}^0, \mathscr{G}^{1:N}).
\label{sigma_independence_1}
\end{align}
\end{small}
\end{claim}

\begin{proof}
Showing the correctness of \eqref{sigma_independence_1} is the same as showing 

\vspace{-4mm}
\begin{small}
\begin{align}
\ee[ \prod_{n \in \mathcal{N}} \ind_{A^n} \vert \mathscr{F}^0, \mathscr{G}^{1:N}]
= \prod_{n \in \mathcal{N}} \ee[\ind_{A^n} \vert \mathscr{F}^0, \mathscr{G}^{1:N}].
\label{sigma_independence_2}
\end{align}
\end{small}
The left hand side of \eqref{sigma_independence_2} can be written as,
\begin{small}
\begin{align}
&\ee[ \prod_{n \in \mathcal{N}} \ind_{A^n}  \vert \mathscr{F}^0, \mathscr{G}^{1:N}] \notag \\
&= \ee \Big[ \ee[\prod_{n \in \mathcal{N}} \ind_{A^n} 
\vert \mathscr{F}^0, \mathscr{G}^k, \{\mathscr{F}^m\}_{ m \in \mathcal{N} \setminus \{k\}}
] \Big \vert \mathscr{F}^0,  \mathscr{G}^{1:N} \Big]
\notag \\
&= \ee \Big[\prod_{n \neq k}\ind_{A^n}  \ee[ \ind_{A^k}
\vert \mathscr{F}^0, \mathscr{G}^k, \{\mathscr{F}^m\}_{m \in \mathcal{N} \setminus \{k\}}
] \Big \vert \mathscr{F}^0, \mathscr{G}^{1:N} \Big]
\notag \\
&= \ee \Big[\prod_{n \neq k}\ind_{A^n}   \ee[ \ind_{A^k} 
\vert \mathscr{F}^0, \mathscr{G}^k
] \Big \vert \mathscr{F}^0, \mathscr{G}^{1:N}\Big]
\notag \\
&= \ee \Big[\prod_{n \neq k}\ind_{A^n}  \Big \vert \mathscr{F}^0, \mathscr{G}^{1:N} \Big]
\ee[ \ind_{A^k} \vert \mathscr{F}^0, \mathscr{G}^k] 
\label{sigma_independence_3}
\end{align}
\end{small}
where the first equality is true due to the tower property of conditional expectation, the second property is true due to ``pulling out known factors" property; the third equality is obtained by first using ``chain rule" property to show that $\mathscr{F}^k$ is conditionally independent of $\{\mathscr{F}^m\}_{ m \in \mathcal{N} \setminus \{k\}}$ given $\mathscr{F}^0, \mathscr{G}^k$ and then using Doob's conditional independence property \cite[Chapter 5]{Kallenberg}; the fourth equality is true again due to ``pulling out known factors" property. 

By repeating the procedure of \eqref{sigma_independence_3} one by one for each $k \in \mathcal{N}$, then we get 

\vspace{-3mm}
\begin{small}
\begin{align}
& \ee \Big[\prod_{n \neq k}\ind_{A^n}  \Big \vert \mathscr{F}^0, \mathscr{G}^{1:N} \Big]
\ee[ \ind_{A^k} \vert \mathscr{F}^0, \mathscr{G}^k]  \notag \\
&= \prod_{k \in \mathcal{N}}  \ee[\ind_{A^k} \vert \mathscr{F}^0, \mathscr{G}^k] 
= \prod_{n \in \mathcal{N}}  \ee[\ind_{A^n} \vert \mathscr{F}^0, \mathscr{G}^{1:N} ]
\end{align}
\end{small}
where last equality is true due to the ``chain rule" property and  Doob's conditional independence property.
\end{proof}
\vspace{-4mm}
\begin{claim}
\label{lm:independe}
\begin{enumerate}
\item Consider feasible strategies $g=g^{0:N}$, $g^n \in \mathcal{G}^n$, $n \in \mathcal{N}$, in Problem \ref{problem1}. 
Then, the random vectors $X_{0:t}^{n}$ are conditionally independent of $X_{0:t}^{m}$ for $n, m \in \mathcal{N}, n \neq m$ given $H^0_t$. 
That is, for any measurable sets $E^n_{0:t} \subset \prod_{s=0}^t \R^{d_X^n}$, $n \in \mathcal{N}$,

\vspace{-2mm}
\begin{small}
\begin{align}
&\prob^{g}( \{X_{0:t}^n \in E_{0:t}^n\}_{ n \in \mathcal{N}}|H^0_{t}) 
= \prod_{n \in \mathcal{N}} \prob^{g}(X_{0:t}^n \in E_{0:t}^n | H^0_{t}).
\label{eq:independence_strategies0T}
\end{align}
\end{small}
\item The same result holds under any feasible fixed prescription strategy $\phi^{prs}\in \Phi^{prs}$ in Problem \ref{problem:equivalent}.
\end{enumerate}
\end{claim}

\begin{proof}
We prove \eqref{eq:independence_strategies0T} by induction. At time $0$, \eqref{eq:independence_strategies0T} is true because random vectors $X_0^{1:N}$ are independent. Suppose \eqref{eq:independence_strategies0T} is true at time $t$. 

At time $t+1$, for all $n\in\mathcal{N}$, define

\vspace{-3mm}
\begin{small}
\begin{align}
&\tilde f^n_t(X_{0:t}^n, W_{t}^n, H_{t}^0)=X_{t+1}^n  \notag \\
&=A^{nn} X_{t}^n + B^{nn} g_{t}^{n}(X_{0:t}^n,H_{t}^0, U_{0:t-1}^n)+ B^{n0} g_{t}^{0}(H_{t}^0)+ W_{t}^n \notag \\
&=A^{nn} X_{t}^n + B^{nn} \check{g}_{t}^{n}(X_{0:t}^n,H_{t}^0)+ B^{n0} g_{t}^{0}(H_{t}^0)+ W_{t}^n
\label{eq:dynamic_function_1}
\end{align}
\end{small}
where $\check{g}_{t}^{n}$ is obtained from $g^n_{0:t}$ by recursively substituting for $U_{0:t-1}^n$.

Then, the left hand side of \eqref{eq:independence_strategies0T} at time $t+1$ becomes

\vspace{-3mm}
\begin{small}
\begin{align}
&\prob^{g} \Big(\{X_{0:t+1}^n \in E_{0:t+1}^n\}_{ n \in \mathcal{N}}|H^0_{t+1} \Big)  
\notag\\
&=\prob^{g} \Big(\{X_{0:t+1}^n \in E_{0:t+1}^n\}_{ n \in \mathcal{N}} \Big \vert H^0_{t},\{\Gamma_{t+1}^{k},\Gamma_{t+1}^{k} X_{t+1}^ {k} \}_{k \in \mathcal{N}}  \Big)  
\notag\\
&=\prob^{g} \Big(\{\tilde f^n_t(X_{0:t}^n, W_{t}^n,H_{t}^0) \in E^n_{t+1},
X_{0:t}^n \in E_{0:t}^n\}_{n \in \mathcal{N}}
\notag \\
&\hspace{1cm}
\Big\vert H^0_{t},
\{\Gamma_{t+1}^{k}, \Gamma_{t+1}^{k} \tilde f^{k}_t(X_{0:t}^{k}, W_{t}^{k},H_{t}^0) \}_{k \in \mathcal{N}}  \Big).
\label{eq:0:T_t_part1}
\end{align}
\end{small}
Note that 1) $\Gamma^{1:N}_{t+1},W_{t}^{1:N}$ are independent of all other variables at time $t$, and 2) $X_{0:t}^{1:N}$ are independent conditioned on $H^0_{t}$ from the induction hypothesis. Hence, if we define $\mathscr{F}^0 = \sigma (H_t^0, \Gamma_{t+1}^{1:N})$, and for $k \in \mathcal{N}$, we define $\mathscr{F}^{k} = \sigma(X_{0:t}^{k}, W_t^{k}, H_t^0, \Gamma_{t+1}^{k})$ and $\mathscr{G}^{k} = \sigma(\Gamma_{t+1}^{k} \tilde f^{k}_t(X_{0:t}^{k}, W_{t}^{k},H_{t}^0))$, then it can be shown that $\mathscr{F}^{1:N}$ are conditionally independent given $\mathscr{F}^0$. 
Then, by using Claim \ref{lm:conditional_independence}, we can write 

\vspace{-3mm}
\begin{small}
\begin{align}
&\prob^{g} \Big(\{\tilde f^n_t(X_{0:t}^n, W_{t}^n,H_{t}^0) \in E^n_{t+1},
X_{0:t}^n \in E_{0:t}^n\}_{ n \in \mathcal{N}}
\notag \\
&\hspace{1cm}
\Big\vert H^0_{0},
\{\Gamma_{t+1}^{k}, \Gamma_{t+1}^{k}\tilde f^{k}_t(X_{0:t}^{k}, W_{t}^{k},H_{t}^0) \}_{k \in \mathcal{N}}  \Big)  \notag \\
& = \prod_{n  \in \mathcal{N}} \prob^{g} \Big(\tilde f^n_t(X_{0:t}^n, W_{t}^n,H_{t}^0) \in E^n_{t+1},
X_{0:t}^n \in E_{0:t}^n
\notag \\
&\hspace{1cm}
\vert H^0_{t},
\{\Gamma_{t+1}^{k}, \Gamma_{t+1}^{k} \tilde f^{k}_t(X_{0:t}^{k}, W_{t}^{k},H_{t}^0) \}_{ k \in \mathcal{N}}  \Big)  \notag \\
& = \prod_{n  \in \mathcal{N}} \prob^{g} \Big(X_{0:t+1}^n \in E_{0:t+1}^n \vert  H^0_{t+1} \Big) .
\label{eq:0:T_t_part1}
\end{align}
\end{small}
Therefore, \eqref{eq:independence_strategies0T} is true at time $t$ and the proof of the first part is complete.

The second part can be proved in a similar way. 
\end{proof}

\vspace{-5mm}
\begin{corollary}
\label{cor_independence}
For any feasible prescription strategy $\phi^{prs} \in \Phi^{prs}$ in Problem \ref{problem:equivalent} ($g=g^{0:N}$, $g^n \in \mathcal{G}^n$, $n \in \mathcal{N}$, in Problem \ref{problem1}), $(X_{t+1}^n,Z_{t+1}^n)$ is conditionally independent of $\{Z_{t+1}^m \}_{ m \in \mathcal{N} \setminus \{n\}}$ given any realization $H_t^0$.
\end{corollary}

\begin{claim}
\label{lm:quadratic_problems}
\begin{enumerate}[1)]
\item For any constant vector $x \in \prod_{n=1}^N \R^{d_X^n}$, 

\vspace{-4mm}
\begin{small}
\begin{align*}
 &\min_{\hspace{10pt} \{u^n \in \R^{d_U^n}\}_{n \in \overline{\mathcal{N}}}}
\hspace{-10pt} QF\left(R_t,\vecc(x,\{u^n\}_{n \in  \overline{\mathcal{N}}})\right) 
= QF\left(P_t, x\right)
\end{align*}
\end{small}
where $P_t:=R^{XX}_t - R^{XU}_t  \left(R^{UU}_t\right)^{-1} R^{UX}_t$ is the Schur complement of $R^{UU}_t$ of $R_t$ and 
the optimal solution is given by, 

\vspace{-5mm}
\begin{small}
\begin{align*}
\vecc(\{u^{n*}\}_{n \in \overline{\mathcal{N}}}) = &-\left(R^{UU}_t\right)^{-1} R^{UX}_tx.
\end{align*}
\end{small}
\item For any $\theta^{1:N} \in \prod_{m=1}^N \Delta(\R^{d_X^m})$, let $X^{\theta^1},\dots,X^{\theta^N}$ be independent random variables such that 
$X^{\theta^n}$ has distribution $\theta^n, n \in \mathcal{N}$. Then

\vspace{-5mm}
\begin{small}
\begin{align}
&\hspace{-20pt} \min_{
\begin{small}
 \begin{array}{c} 
\{q^n\in\mathcal{Q}^n(\theta^n)\}_{n \in \mathcal{N}}  \end{array}
 \end{small} } \notag \\
& \tr \bigg( R_t \cov \Big( \vecc(\{X^{\theta^{n}}\}_{n \in \mathcal{N}}, 0, \{q^n(X^{\theta^n})\}_{n \in \mathcal{N}}) \Big) \bigg) 
\nonumber\\
&= \sum_{n=1}^N \tr \Big( P_t^{nn} \cov(X^{\theta^n}) \Big)
\label{eq:QP2}
\end{align}
\end{small}
where 
\begin{align}
P_t^{nn}:=R^{X^nX^n}_t - R^{X^nU^n}_t  \left(R^{U^nU^n}_t\right)^{-1} R^{U^nX^n}_t
\label{eq:QP2_sol_matrix}
\end{align}
and  the optimal solution for $n \in \mathcal{N}$ is given by, 
\begin{small}
\begin{align}
\hspace{-3mm} q^{n*}(X^{\theta^n}) = -\left(R^{U^nU^n}_t\right)^{-1} R^{U^nX^n}_t\left(X^{\theta^n}-\mu(\theta^n)\right).
\label{eq:QP2_sol}
\end{align}
\end{small}
\end{enumerate}
\end{claim}

\begin{proof}
The first part of Claim \ref{lm:quadratic_problems} can be obtained by a simple completing the square argument.

Now let's consider the functional optimization problem \eqref{eq:QP2} in the second part of Claim  \ref{lm:quadratic_problems}. 
Using properties of trace and covariance matrices, we can write

\vspace{-4mm}
\begin{small}
\begin{align}
& \tr \bigg( R_t \cov \Big( \vecc(\{X^{\theta^{n}}\}_{n \in \mathcal{N}}, 0, \{q^n(X^{\theta^n})\}_{n \in \mathcal{N}}) \Big) \bigg) 
\notag \\
& = \ee\Big[ QF \Big(R_t, \vecc(\{X^{\theta^{n}}\}_{n \in \mathcal{N}}, 0, \{q^n(X^{\theta^n}) \}_{n \in \mathcal{N}} ) 
\notag \\
&- \ee[ \vecc (\{X^{\theta^{n}}\}_{n \in \mathcal{N}}, 0 ,\{q^n(X^{\theta^n}) \}_{n \in \mathcal{N}} )]  \Big) \Big]
\notag \\
& = \ee\Big[ QF \Big(R_t, \vecc( \{X^{\theta^n} - \mu(\theta^n) \}_{n \in \mathcal{N}}, 0, \{q^n(X^{\theta^n}) \}_{n \in \mathcal{N}} ) \Big) \Big]
\notag \\
& = \sum_{n \in \mathcal{N}}  \ee \Big[ QF \Big( \tilde{R}_t^n, \vecc \Big(X^{\theta^n}-\mu(\theta^n),  q^n(X^{\theta^n}) \Big) \Big)  \Big]
\label{eq:QP2_obj1}
\end{align}
\end{small}
where $\tilde R_t^n=\begin{bmatrix}
 R_t^{X^nX^n} & R_t^{X^nU^n} \\
 R_t^{U^nX^n} & R_t^{U^nU^n}
\end{bmatrix}$.
The last equality in \eqref{eq:QP2_obj1} is true because all off-diagonal terms are zero since $\ee[q^n(X^{\theta^n})]=0, \ee[X^{\theta^n}-\mu(\theta^n)]=0$, and $X^{\theta^{n}}$ and $X^{\theta^{m}}$ are independent for all $n\neq m$. 

Note that each term in \eqref{eq:QP2_obj1} only depends on one $q^n$, $n \in \mathcal{N}$. Therefore, the functional optimization problem \eqref{eq:QP2} is equivalent to solving the $N$ optimization problems

\vspace*{-3mm}
\begin{small}
\begin{align}
\min_{q^n\in\mathcal{Q}^{n}(\theta^n)}
\ee \Big[ QF \Big( \tilde{R}_t^n, \vecc \Big(X^{\theta^n}-\mu(\theta^n),  q^n(X^{\theta^n}) \Big) \Big)  \Big].
\label{eq:QP2_obj22}
\end{align}
\end{small}
Since $\theta^n$ is the distribution of $X^{\theta^n}$, we have

\vspace{-3mm}
\begin{small}
\begin{align}
&\ee\left[
QF\left(
\tilde{R}_t^n, \vecc\left(X^{\theta^n}-\mu(\theta^n),q^n(X^{\theta^n})\right)
\right)
\right]
\nonumber\\
=&\int
QF\left(
\tilde{R}_t^n, \vecc\left(y-\mu(\theta^n),q^n(y)\right)
\right)
\theta^n(dy).
\label{eq:OP2_obj2}
\end{align}
\end{small}
Note that the function inside the integral of \eqref{eq:OP2_obj2} is a quadratic function.
As in the first part of Claim \ref{lm:quadratic_problems}, for any $y\in\R^{d_X^n}$ we have

\vspace{-4mm}
\begin{small}
\begin{align*}
&QF\left(
\tilde{R}_t^n, \vecc\left(y-\mu(\theta^n),q^n(y)\right)
\right) \geq
\nonumber\\ 
&QF\left(
\tilde{R}_t^n, \vecc\left(y-\mu(\theta^n),q^{n*}(y)\right)
\right)
=
QF\left(P_t^{nn},y-\mu(\theta^n)\right)
\end{align*}
\end{small}
where $P_t^{nn}$ is given by \eqref{eq:QP2_sol_matrix} and $q^{n*}$ is the function given by \eqref{eq:QP2_sol} .
Note that $q^{n*}\in\mathcal{Q}^{\theta^n}$ because $q^{n*}$ is measurable and 

\vspace{-3mm}
\begin{small}
\begin{align*}
&\int q^{n*}(x^n)\theta^n(dx^n)
= \notag \\
&\int
-\left(R_t^{U^n U^n}\right)^{-1} R_t^{U^n X^n}
\left(x^n-\mu(\theta^n)\right)
\theta^n(dx^n)
=0.
\end{align*}
\end{small}
Thus $q^{n*}$ is the optimal solution for the optimization problem in \eqref{eq:QP2_obj22} for each $n \in \mathcal{N}$
and the optimal value is $\tr( P_t^{nn} \cov(X^{\theta^n}))$. 
Using \eqref{eq:QP2_obj1}, the optimal value in \eqref{eq:QP2} then becomes
$\sum_{n=1}^N \tr\left(P_t^{nn} 
\cov(X^{\theta^n})
\right).$
\end{proof}
\begin{remark}
The functional optimization in part 2 of Claim 3 can be thought of as a static team problem \cite{radner1962team} where players are constrained to use zero-mean strategies.
\end{remark}

\vspace{-4mm}
\section{Proof of Lemma \ref{lm:first_structure}}
\label{Proof_Lm_first_structure}

To show that the local controller $C^{n}$, $n \in \mathcal{N}$, can use only $\hat H^{n}_t = \{X_{t}^n, H^0_t\}$ to make the decision at time $t$ without loss of optimality, we proceed using \emph{person-by-person} approach. 

For any fixed feasible strategies of the remote controller $g^{0}$ and the local controllers $g^{m}$, $m \in \mathcal{N} \setminus \{n\}$, the problem of finding optimal strategy of the local controller $n$ becomes a centralized problem with  the state $\hat X_t = \vecc(X_{0:t-1}^{-n}, X_t^{1:N}, H_t^0)$. From the theory of centralized control problems with imperfect information \cite{bertsekas1978stochastic}, we know
that we can restrict controller $C^n$'s strategy to be of the form:
\begin{align}
U_t^{n*} = \sigma_t^{n}(\prob^{g^{0},g^{-n}}(\hat X_t \vert H_t^{n}) ).
\end{align}
Then, if we denote $\tilde g=  (g^{0},g^{-n})$, for any measurable sets $F \subset \mathcal{H}^{0}_t$, $E_s^m \in \R^{d_X^m}$, $m \in \mathcal{N} \setminus \{n\}$, $s=0,1,\ldots,t-1$, and $E_t^m \in \R^{d_X^m}$, $m \in \mathcal{N}$,

\vspace{-3mm}
\begin{small}
\begin{align}
&\prob^{\tilde g}\Big(\hat X_t \in  E_{0:t-1}^{-n} \times E_t^{1:N} \times F \vert H_t^{n}\Big) \notag\\
&= 
\prob^{\tilde g}\big( \vecc(X_{0:t}^{-n} ) \in E_{0:t}^{-n}, X_t^n \in E_t^n, H_t^0 \in F  \vert  X_{0:t}^n, H_t^0 \big) \notag\\
&= \mathds{1}_{F}(H_t^0)\prob^{\tilde g}\big(\vecc(X_{0:t}^{-n} ) \in E_{0:t}^{-n}, X_t^n \in E_t^n  \vert  X_{0:t}^n, H_t^0 \big)
\notag \\
&= \mathds{1}_{F}(H_t^0) \prob^{\tilde g}(X_t^n \in E_t^n \vert X_{t}^n, H_t^0) \prod_{m \in \mathcal{N} \setminus \{n\}} \prob^{\tilde g}(X_{0:t}^m \in E_{0:t}^m \vert H_t^0)
\notag \\
&=\prob^{\tilde g}\Big(\hat X_t \in E_{0:t-1}^{-n} \times E_t^{1:N} \times F \vert X_t^n, H_t^{0}\Big)
\label{eq:first_structure_proof}
\end{align}
\end{small}
where the second equality is true due to the ``pulling out known factors" property, the third equality is true from Claim \ref{lm:independe}, and the last equality follows from the same reasons as the first three equalities.
Therefore, the local controller $C^{n}$ can use only $\hat H^{n}_t = \{X_{t}^n, H^0_t\}$ to make the optimal decision at time $t$. 

\textit{An alternative proof based on Markov Decision Process (MDP):}\\
To show that the local controller $C^{n}$, $n \in \mathcal{N}$, can use only $\hat H^{n}_t = \{X_{t}^n, H^0_t\}$ to make the decision at time $t$ without loss of optimality, we proceed using \emph{person-by-person} approach. 
For any fixed feasible strategies of the remote controller $g^{0}$ and the local controllers $g^{m}$, $m \in \mathcal{N} \setminus \{n\}$, the problem of finding optimal strategy of the local controller $n$ can be reduced to a MDP with $(X^n_t, H^0_t)$ as the (perfectly observed) state---In particular, it can be shown that this state evolves in a controlled Markovian fashion with $U^n_t$ as the control action. Moreover,  by averaging over $X^{-n}_{0:t}$, the expected cost at time $t$ can be written as a function of this state and the action $U^n_t$. From the theory of centralized control problems with perfect state information \cite{bertsekas1978stochastic}, we know that we can restrict to control strategies for $C^n$ that are of the form:
\begin{align}
U_t^{n*} = \sigma_t^{n}(X^n_t, H^0_t).
\end{align}
\vspace{-4mm}
\section{Proof of Lemma \ref{lm:beliefupdate}}
\label{Proof_Lm_belief_update}
Note that from \eqref{eq:theta0}-\eqref{eq:dynamic_function}, $[\nu_t^n (\cdot)](E^n): \mathcal{H}_t^0  \to \R$ is a measurable function. To show that
$[\nu_{t}^n(H_{t}^0)] (E^n) = \prob^{\phi^{prs}_{0:t-1}}(X_t^n \in E^n | H^0_{t})$,
first note that for any $t$

\vspace{-4mm}
\begin{small}
\begin{align}
&\prob^{\phi^{prs}_{0:t-1}}(X_{t}^n \in E^n | H^0_{t}) = \prob^{\phi^{prs}_{0:t-1}}(X_{t}^n \in E^n | H^0_{t-1}, Z_{t}^{1:N}) \notag \\
&= \prob^{\phi^{prs}_{0:t-1}}(X_{t}^n \in E^n | H^0_{t-1}, Z_{t}^n),
\label{eq:nu_t_proof}
\end{align}
\end{small}
where the second equality is true because of Corollary \ref{cor_independence}.

We now prove by induction that 
\begin{align}
[\nu_{t}^n(H_{t}^0)] (E^n) = \prob^{\phi^{prs}_{0:t-1}}(X_{t}^n \in E^n | H^0_{t-1}, Z_{t}^n).
\label{nu_claim}
\end{align}
At time $t=0$, since $\Gamma_0^n \in \{0,1\}$, consider two cases:
\begin{itemize}
\item If $\Gamma_0^n =1$,
\begin{small}
\begin{align}
&\prob(X_0^n \in E^n | Z_0^{n})\ind_{\{\Gamma_0^n =1\}} = \prob(X_0^n \in E^n | X_0^n, \Gamma_0^n) \ind_{\{\Gamma_0^n =1\}} \notag \\
&= \prob(X_0^n \in E^n | X_0^n) \ind_{\{\Gamma_0^n =1\}} = \ind_{E^n}(X_0^n) \ind_{\{\Gamma_0^n =1\}}.
\label{eq:gamma1_t0}
\end{align}
\end{small}
\item If $\Gamma_0^n =0$,
\begin{small}
\begin{align}
&\prob(X_0^n \in E^n | Z_0^{n}) \ind_{\{\Gamma_0^n =0\}} = \prob(X_0^n \in E^n | \Gamma_0^n) \ind_{\{\Gamma_0^n =0\}}  \notag \\
&= \prob(X_0^n \in E^n) \ind_{\{\Gamma_0^n =0\}} = \pi_{X_0^n}(E^n) \ind_{\{\Gamma_0^n =0\}}. 
\end{align}
\end{small}
\end{itemize}
Hence, \eqref{nu_claim} holds at time $0$. Assume that \eqref{nu_claim} holds at time $t$.
This means that $\prob^{\phi^{prs}_{0:t-1}}(dx_t^n | H^0_{t}) = [\nu_{t}^n(H_{t}^0)] (dx_t^n)$ and since $W_t^n$ is independent of all random variables at and before time $t$, we get
\begin{align}
\prob^{\phi^{prs}_{0:t-1}}(dx_t^n dw_t^n| H^0_{t}) 
= [\nu_{t}^n(H_{t}^0)] (dx_t^n) \pi_{W_t^n}(dw_t^n).
\label{eq:nu_induction}
\end{align}
At time $t+1$, since $\Gamma_{t+1}^n \in \{0,1\}$, consider two cases:
\begin{itemize}
\item If $\Gamma_{t+1}^n =1$, similar to \eqref{eq:gamma1_t0} we obtain
\begin{small}
\begin{align}
&\hspace{-6mm}\prob^{\phi^{prs}_{0:t}}(X_{t+1}^n \in E^n | H^0_t, Z_{t+1}^n) 
\ind_{\{\Gamma_{t+1}^n =1\}} \notag \\
&\hspace{-6mm}=\ind_{E^n}(X_{t+1}^n) \ind_{\{\Gamma_{t+1}^n =1\}}
=[\nu_{t+1}^n(H_{t+1}^0)](E^n) \ind_{\{\Gamma_{t+1}^n =1\}}.
\end{align}
\end{small}
\vspace{-5mm}
\item If $\Gamma_{t+1}^n =0$,
\begin{small}
\begin{align}
 &\hspace{0mm} \prob^{\phi^{prs}_{0:t}}(X_{t+1}^n \in E^n | H^0_t, Z_{t+1}^n) \ind_{\{\Gamma_{t+1}^n =0\}} \notag \\
 &\hspace{0mm}= \prob^{\phi^{prs}_{0:t}}(X_{t+1}^n \in E^n | H^0_t, \Gamma_{t+1}^n) \ind_{\{\Gamma_{t+1}^n =0\}} \notag \\
&\hspace{0mm}=\prob^{\phi^{prs}_{0:t}}(f_{t+1}^n(X_t^n,W_t^n,\phi_t^{prs}(H^0_t)) \in E^n | H^0_t) \ind_{\{\Gamma_{t+1}^n =0\}} \notag \\
&\hspace{0mm}=\ee^{\phi^{prs}_{0:t}}[\mathds{1}_{E^n}(f_{t+1}^n(X_t^n,W_t^n,\phi_t^{prs}(H^0_t))) | H^0_t] \ind_{\{\Gamma_{t+1}^n =0\}} \notag \\
&\hspace{0mm}=\int\int\mathds{1}_{E^n}(f_{t+1}^n(x_t^n,w_t^n,\phi_t^{prs}(H^0_t))) 
\notag \\
&\hspace{13mm}\prob^{\phi^{prs}_{0:t-1}}(dx_t^n dw_t^n \vert H_t^0) \ind_{\{\Gamma_{t+1}^n =0\}} \notag \\
&\hspace{0mm}=\int\int\mathds{1}_{E^n}(f_{t+1}^n(x_t^n,w_t^n,\phi_t^{prs}(H^0_t))) \notag \\
&\hspace{13mm} \nu_{t}^n(H_{t}^0) (dx_t^n) \pi_{W_t^n}(dw_t^n) \ind_{\{\Gamma_{t+1}^n =0\}} \notag \\
&=[\nu_{t+1}^n(H_{t+1}^0)](E^n) \ind_{\{\Gamma_{t+1}^n =0\}},
\end{align}
\end{small}
where the second equality is true due to \eqref{eq:dynamic_function} and the fact that $\Gamma_{t+1}^n$ is independent of $X_{t+1}^n$ and $H^0_t$, the fourth equality is true due to the disintegration theorem \cite{Kallenberg}, and the fifth equality is true due to  \eqref{eq:nu_induction}.
\end{itemize}
Hence, \eqref{nu_claim} holds at time $t+1$ and the proof of Lemma \ref{lm:beliefupdate} is complete.
\vspace{-7mm}
\section{Proof of Theorem \ref{thm:structure}}
\label{Proof_Thm_structure}
For any $\phi^{prs} \in \Phi^{prs}$ and any realization $h^0_t\in\mathcal{H}^0_t$, let the realization of the common belief $\Theta_t^n$ be
$\theta_t^n=\nu_t^{n}(h_t^0)$, $n \in \mathcal{N}$, defined by Lemma \ref{lm:beliefupdate}.
Suppose the prescription strategy $\phi^{prs*} \in \Phi^{prs}$ achieves the minimum of \eqref{eq:DP_V} for $\theta_t^{1:N}$, $t=0,\ldots,T$, and let $u^{prs*}_t=(u^{0*}_t,\bar u^{1:N*}_t,q_t^{1:N*}) = \phi^{prs*}(h_t^0)$ for any realization $h^0_t\in\mathcal{H}^0_t$.

We prove by induction that $V_{t}\big(\{ \nu_{t}^n(h_{t}^0) \}\big)$ is a measurable function with respect to $h_t^0$, and for any $h_t^0 \in \mathcal{H}_t^0$ we have
\begin{align}
 &\ee^{\phi'_t}\left[\sum_{s=t}^T  c_s^{prs}(X^{1:N}_s,U^{prs}_s)\middle| h_t^0 \right]
 \nonumber\\
&= V_{t}\big(\{ \nu_{t}^n(h_{t}^0) \}_{n \in \mathcal{N}}  \big)
\label{eq:Vinduction_part1}
\\
& \leq  \ee^{\phi^{prs}}\left[\sum_{s=t}^T c_s^{prs}(X^{1:N}_s,U^{prs}_s) \middle| h_t^0\right]
\label{eq:Vinduction_part2}
\end{align}
where $\phi'_t= \{\phi^{prs}_{0:t-1},  \phi^{prs*}_{t:T}\}$.
Note that the above equation at $t=0$ gives the optimality of $\phi^{prs*}$ for Problem \ref{problem:equivalent}.

We first consider \eqref{eq:Vinduction_part1}.
At $T+1$, \eqref{eq:Vinduction_part1} is true (all terms are defined to be $0$ at $T+1$). 
Assume $V_{t+1}\big(\{ \nu_{t+1}^n(h_{t+1}^0)\}_{n \in \mathcal{N}}\big)$
is a measurable function with respect to $h_{t+1}^0$ and
\eqref{eq:Vinduction_part1} is true at $t+1$.

From the tower property of conditional expectation we have

\vspace{-3mm}
\begin{small}
\begin{align}
&\ee^{\phi'_t}\left[\sum_{s=t}^T  c_s^{prs}(X^{1:N}_s,U^{prs}_s)  \middle| h_t^0 \right]
\nonumber\\
&=\ee^{\phi'_t}\left[ c_t^{prs}(X^{1:N}_t,U^{prs}_t)  \middle| h_t^0 \right] \nonumber\\
&+\ee^{\phi'_t}\left[\ee^{\phi'_t}\left[\sum_{s=t+1}^T c_s^{prs}(X^{1:N}_s,U^{prs}_s)  \middle| H_{t+1}^0 \right]\middle| h_t^0 \right].
\label{eq:Vinduction_part1_zero}
\end{align}
\end{small}
Note that the first term in \eqref{eq:Vinduction_part1_zero} is equal to

\vspace{-3mm}
\begin{small}
\begin{align}
\int  c_t^{prs} (x_t^{1:N}, u_t^{prs*}) \prod_{n\in\mathcal{N}}\theta_t^n(dx_t^n) = \textit{IC}(\theta_t^{1:N},u_t^{prs*}).
\label{eq:Vinduction_part1_first} 
\end{align}
\end{small}
From the induction hypothesis, $V_{t+1}\big(\{ \nu_{t+1}^n(h_{t+1}^0)\}_{n \in \mathcal{N}}\big)$ is measurable with respect to $h^0_{t+1}$, and \eqref{eq:Vinduction_part1} holds at $t+1$. Since $\nu_{t+1}^n(h_{t+1}^0) =  \psi_t^n(\theta_t^n,u^{prs*}_t,z^n_{t+1}) $, the second term in \eqref{eq:Vinduction_part1_zero} can be written as

\vspace{-3mm}
\begin{small}
\begin{align}
&\ee^{\phi'_t} \Big[V_{t+1} \Big( \Big\{ \nu_{t+1}^n(H_{t+1}^0) \Big\}_{n \in \mathcal{N}} \Big) \Big\vert  h_t^0\Big] \notag \\
&= \ee^{\phi'_t} \Big[V_{t+1} \Big( \Big\{ \psi_t^n(\theta_t^n,u^{prs*}_t,Z^n_{t+1}) \Big\}_{n \in \mathcal{N}} \Big) \Big\vert h_t^0\Big]
\nonumber\\
 &= \sum_{\gamma_{t+1}^1 \in\{0,1\}} \ldots \sum_{\gamma_{t+1}^N \in\{0,1\}} 
 \Big[ \prod_{n\in\mathcal{N}}  \textit{LS}(p_n, \gamma_{t+1}^n) \Big] \times
 \nonumber\\
&\ee^{\phi'_t}\Big[V_{t+1} \Big( \Big\{ \psi_t^n(\theta_t^n,u^{prs*}_t,Z^n_{t+1}) \Big\}_{n \in \mathcal{N}} \Big) \Big\vert h_t^0, \{\Gamma_{t+1}^n=\gamma_{t+1}^n \}_{n \in \mathcal{N}} \Big]
\label{eq:Vinduction_part1_second_p1} 
\end{align}
\end{small}
where \eqref{eq:Vinduction_part1_second_p1} follows from the fact that $\prob(\Gamma_{t+1}^n=0) = 1-\prob(\Gamma_{t+1}^n=1)=p_n$.
From Lemma \ref{lm:beliefupdate} we have

\vspace{-3mm}
\begin{small}
\begin{align}
 \psi_t^n(\theta_t^n,u^{prs*}_t,Z^n_{t+1})
&=(1- \Gamma_{t+1}^n)\alpha_t^{n*} + \Gamma^n_{t+1} \rho(X^n_{t+1})\notag \\
&=
\textit{NB}(\Gamma_{t+1}^n, \alpha_t^{n*}, X_{t+1}^n) 
\end{align}
\end{small}
where $\alpha_t^{n*} = \psi_t^n(\theta_t^n,u^{prs*}_t,\emptyset)$.
Consequently, each inner term in \eqref{eq:Vinduction_part1_second_p1} can be written as

\vspace{-3mm}
\begin{small}
\begin{align}
&\ee^{\phi'_t}\Big[V_{t+1} \Big( \Big\{ 
\textit{NB}(\Gamma_{t+1}^n, \alpha_t^{n*}, X_{t+1}^n) 
  \Big\}_{n \in \mathcal{N}} \Big) \Big\vert h_t^0, \{\Gamma_{t+1}^n=\gamma_{t+1}^n \}_{n \in \mathcal{N}} \Big] 
\nonumber\\
&=\ee^{\phi'_t}\Big[V_{t+1} \Big( \Big\{
\textit{NB}(\gamma_{t+1}^n, \alpha_t^{n*}, X_{t+1}^n) 
  \Big\}_{n \in \mathcal{N}} \Big) \Big\vert h_t^0, \{\Gamma_{t+1}^n=0 \}_{n \in \mathcal{N}} \Big] 
\nonumber\\
&= \int
V_{t+1} \Big( \Big\{ 
\textit{NB}(\gamma_{t+1}^n, \alpha_t^{n*}, x_{t+1}^n) 
  \Big\}_{n \in \mathcal{N}} \Big)
\prod_{n\in\mathcal{N}} \alpha_t^n(dx_{t+1}^n).
\label{eq:Vinduction_part1_second_inner} 
\end{align}
\end{small}
The first equality in \eqref{eq:Vinduction_part1_second_inner} is true because $X_{t+1}^{1:N}$ are independent of $\Gamma_{t+1}^{1:N}$ and
the last equality in \eqref{eq:Vinduction_part1_second_inner}  follows from Lemma \ref{lm:beliefupdate}.

Combining \eqref{eq:Vinduction_part1_first}, \eqref{eq:Vinduction_part1_second_p1}, and \eqref{eq:Vinduction_part1_second_inner}, the right hand side of \eqref{eq:Vinduction_part1_zero} is $V_{t}\big( \theta_t^{1:N} \big)$ from the definition of the value function \eqref{eq:DP_V} which is equal to $V_{t}\big(\{ \nu_{t}^n(h_{t}^0) \}_{n \in \mathcal{N}}\big)$. Hence, \eqref{eq:Vinduction_part1} is true at time $t$.
The measurability  of $V_{t}\big(\{ \nu_{t}^n(h_{t}^0) \}_{n \in \mathcal{N}}\big)$ with respect to $h_t^0$ is also resulted from the fact that 
$V_{t}\big(\{ \nu_{t}^n(h_{t}^0) \}_{n \in \mathcal{N}} \big)$ is equal to the conditional expectation $\ee^{\phi'_t}\left[\sum_{s=t}^T c_s^{prs}(X^{1:N}_s,U^{prs}_s) \middle| h^0_t\right]$ which is measurable with respect to $h^0_t$.

Now let's consider \eqref{eq:Vinduction_part2}.
At $T+1$, \eqref{eq:Vinduction_part2} is true (all terms are defined to be $0$ at $T+1$). 
Assume \eqref{eq:Vinduction_part2} is true at $t+1$.
Let $u^{prs}_t=(u^{0}_t,\bar u^{1:N}_t,q_t^{1:N}) = \phi^{prs}(h^0_t)$. 
Following an argument similar to that of \eqref{eq:Vinduction_part1_zero}-\eqref{eq:Vinduction_part1_second_inner},

\vspace{-3mm}
\begin{small}
\begin{align*}
&\ee^{\phi^{prs}}\left[\sum_{s=t}^T c_s(X^{1:N}_s,U^{0:N}_s)  \middle| h_t^0 \right]
\notag\\
&\geq 
 \textit{IC}(\theta_t^{1:N},u_t^{prs})
+
\sum_{\gamma_{t+1}^1 \in\{0,1\}} \ldots \sum_{\gamma_{t+1}^n \in\{0,1\}}   \Big[
  \prod_{n\in\mathcal{N}} \textit{LS}(p_n, \gamma_{t+1}^n) \Big] \times
\notag\\
&\hspace{1mm}\int  V_{t+1} \Big( \Big\{ 
\textit{NB}(\gamma_{t+1}^n, \alpha_t^{n}, x_{t+1}^n)
  \Big\}_{n \in \mathcal{N}} \Big)
\prod_{n\in\mathcal{N}}\alpha_t^n(dx_{t+1}^n) 
\geq V_t(\theta_t^{1:N}).
\end{align*}
\end{small}
where the last inequality follows from the definition of the value function \eqref{eq:DP_V}.
This completes the proof of the induction step, and the proof of the theorem.
\vspace{-2mm}
\section{Proof of Theorem \ref{thm:Sol_packetdrop}}
\label{Proof_Thm_optimal_solution}
The proof is done by induction.  First note that \eqref{eq:Vt_PacketDrop} is true for $t=T+1$ since $P_{T+1} = \mathbf{0}$, $\tilde P_{T+1}^{nn} = \mathbf{0}$ for all $n \in \mathcal{N}$ and, by definition, $e_{T+1}=0$.
Now, suppose \eqref{eq:Vt_PacketDrop} is true at $t+1$ and the matrices $P_{t+1}$ and $\tilde P_{t+1}^{nn}$, for all $n \in \mathcal{N}$, are all PSD.
Let's  compute the right hand side of  \eqref{eq:DP_V} in Theorem \ref{thm:structure}.
We first consider $V_{t+1}$ term on the right hand side in \eqref{eq:DP_V}. From the induction hypothesis we have

\vspace{-3mm}
\begin{small}
\begin{align} 
 &V_{t+1} \Big( \Big\{ 
\textit{NB}(\gamma_{t+1}^n, \alpha_t^n, x_{t+1}^n)
  \Big\}_{n \in \mathcal{N}} \Big)
\notag \\
&= QF \bigg( P_{t+1}, \vecc \Big( \Big\{ \mu \Big(\textit{NB}(\gamma_{t+1}^n, \alpha_t^n, x_{t+1}^n) \Big) \Big)  \Big\}_{n \in \mathcal{N}} \Big) \bigg)
\notag \\
&+ \sum_{n \in \mathcal{N}} \tr \bigg( \tilde P_{t+1}^{nn} \cov \Big( \textit{NB}(\gamma_{t+1}^n, \alpha_t^n, x_{t+1}^n)\Big) \Big) + e_{t+1}
\notag \\
&= QF \bigg( P_{t+1}, \vecc \Big( \Big\{ \mu(\alpha_t^n) + \gamma_{t+1}^n ( x^n_{t+1}-\mu(\alpha_t^n))\Big\}_{n \in \mathcal{N}}  \Big) \bigg)
\notag \\
&+ \sum_{n \in \mathcal{N}} (1-\gamma^n_{t+1})\tr \bigg( \tilde P_{t+1}^{nn} \cov \Big(  \alpha_t^n   \Big) \Big)  \bigg) + e_{t+1}
\label{eq:Vterm}
\end{align}
\end{small}
where the last equality in \eqref{eq:Vterm} is true because $\textit{NB}(\gamma_{t+1}^n, \alpha_t^n, x_{t+1}^n)= (1-\gamma_{t+1}^n)\alpha_t^n + \gamma_{t+1}^n \rho(x_{t+1}^n)$, $\mu( \rho(x^n_{t+1})) = x^n_{t+1}$, and $\cov( \rho(x^n_{t+1})) = 0$.

The first term on the right hand side of \eqref{eq:Vterm} can be further decomposed into

\vspace{-3mm}
\begin{small}
\begin{align} 
&QF \bigg( P_{t+1}, \vecc \Big( \Big\{ \mu(\alpha_t^n) + \gamma_{t+1}^n ( x^n_{t+1}-\mu(\alpha_t^n))\Big\}_{n \in \mathcal{N}}
\Big) \bigg)
\notag \\
&= QF \bigg( P_{t+1}, \vecc \Big( \Big\{ \mu(\alpha_t^n) \Big\}_{n \in \mathcal{N}} \Big) \bigg) +
\notag \\
& 2 \vecc \Big( \Big\{  \mu(\alpha_t^n) \Big\}_{n \in \mathcal{N}} \Big)^{\tp} P_{t+1}
\vecc \Big( \Big\{\gamma^n_{t+1}( x^n_{t+1}-\mu(\alpha_t^n))\Big\}_{n \in \mathcal{N}} \Big)
\notag \\
&+  QF \bigg( P_{t+1}, \vecc \Big( \Big\{\gamma^n_{t+1} (x^n_{t+1}-\mu(\alpha_t^n))\Big\}_{n \in \mathcal{N}} \Big) \bigg).
\label{eq:Vterm_3}
\end{align}
\end{small}
Note that $\int (x^n_{t+1} - \mu(\alpha_t^n) )  \alpha_t^n(dx_{t+1}^n) =0, \hspace{10pt} \forall n \in \mathcal{N}$ and 

$\int (x^n_{t+1} - \mu(\alpha_t^n) )  (x^m_{t+1} - \mu(\alpha_t^m) ) \alpha_t^n(dx_{t+1}^n)
\alpha_t^m(dx_{t+1}^m)  =0, \hspace{10pt} \forall n \neq m$.
Consequently, integrating the right hand side of \eqref{eq:Vterm} with respect to $\prod_{n\in\mathcal{N}}\alpha_t^n(dx_{t+1}^n)$ we get

\vspace{-3mm}
\begin{small}
\begin{align}
&\int  V_{t+1} \Big( \Big\{ 
\textit{NB}(\gamma_{t+1}^n, \alpha_t^n, x_{t+1}^n)
  \Big\}_{n \in \mathcal{N}} \Big) \prod_{n\in\mathcal{N}}\alpha_t^n(dx_{t+1}^n)=
\notag\\
&QF \bigg( P_{t+1}, \vecc \Big( \Big\{  \mu(\alpha_t^n) \Big\}_{n \in \mathcal{N}}  \Big) \bigg)  
+  \sum_{n \in \mathcal{N}}  \gamma^n_{t+1} \tr \Big(P_{t+1}^{nn} \cov(\alpha_t^n) \Big)  
\notag\\
&+ \sum_{n \in \mathcal{N}}  (1-\gamma^n_{t+1}) \tr \Big( \tilde P_{t+1}^{nn} \cov(\alpha_t^n) \Big) +e_{t+1}.
\label{eq:Vterm_4}
\end{align}
\end{small}
Substituting \eqref{eq:Vterm_4} back into \eqref{eq:DP_V}, the second term (the term after $+$) on the right hand side of  \eqref{eq:DP_V} can be written as

\vspace{-3mm}
\begin{small}
\begin{align}
&QF \bigg( P_{t+1}, \vecc \Big( \Big\{  \mu(\alpha_t^n) \Big\}_{n \in \mathcal{N}} \Big) \bigg)   \notag \\
&+  \sum_{n \in \mathcal{N}}  (1-p^n) \tr \Big(P_{t+1}^{nn} \cov(\alpha_t^n) \Big)  
\notag\\
&+ \sum_{n \in \mathcal{N}}  p^n \tr \Big( \tilde P_{t+1}^{nn} \cov(\alpha_t^n) \Big)
+ e_{t+1}.
\label{eq:Vterm_5}
\end{align}
\end{small}

Let $S_t^{\theta_t} := \vecc(\{X^{\theta_t^{n}}\}_{n \in \mathcal{N}},u^0_t, \{\bar u^{n}_t+q_t^n(X^{\theta_t^n}) \}_{n \in \mathcal{N}})$ where 
$X^{\theta_t^n}$ is a random vector with distribution $\theta_t^n$ and $\{X^{\theta_t^{n}}\}_{n \in \mathcal{N}}$ and $W_t^{1:N}$ are independent.
Let $Y_t^{\theta_t^n}$ be the random vector defined by
\begin{align}
Y_t^{\theta_t^n}:&= \begin{bmatrix}
 A & B
\end{bmatrix}_{n \bullet} S_t^{\theta_t} +W_t^n \nonumber\\
= &A^{nn} X^{\theta_t^n} + B^{nn}(\bar u^{n}_t+q_t^n(X^{\theta_t^n}))+B^{n0} u^0_t +W_t^n.
\nonumber
\end{align}
From \eqref{eq:psit} in Lemma \ref{lm:beliefupdate} we know that $Y_t^{\theta_t^n}$ has distribution $\alpha_t^n$ for all $n \in \mathcal{N}$.
Then, \eqref{eq:Vterm_5} becomes

\begin{small}
\begin{align}
&QF \Big(P_{t+1}, \vecc(\{\ee[Y_t^{\theta_t^n}] \}_{n \in \mathcal{N}} ) \Big) 
\notag \\
&+  \sum_{n \in \mathcal{N}}  (1-p^n) \tr \Big(P_{t+1}^{nn} \cov(Y_t^{\theta_t^n}) \Big)  
\notag \\
&+ \sum_{n \in \mathcal{N}}  p^n \tr \Big( \tilde P_{t+1}^{nn} \cov(Y_t^{\theta_t^n}) \Big) + e_{t+1}
\notag \\ 
&=QF \Big(L_{t}, \ee[S_t^{\theta_t}] \Big) +\sum_{n \in \mathcal{N}} \tr \Big( ((1-p^n) \hat L_{t}^{nn}+p^n \tilde L_{t}^{nn}) \cov(S_t^{\theta_t}) \Big )
\notag \\
&+ \sum_{n \in \mathcal{N}}
\tr \Big( ((1-p^n) P_{t+1}^{nn}+p^n \tilde P_{t+1}^{nn} ) \cov(\pi_{W_t^n}) \Big)  + e_{t+1}
\notag \\
&= QF \Big(L_{t}, \ee[S_t^{\theta_t}] \Big) \notag \\
&+\sum_{n \in \mathcal{N}} \tr \Big( ((1-p^n) \hat L_{t}^{nn}+p^n \tilde L_{t}^{nn}) \cov(S_t^{\theta_t}) \Big ) + e_{t},
\label{eq:Vterm_6}
\end{align}
\end{small}
where we have defined 
\begin{align}
\label{def:L}
L_t &= \begin{bmatrix}
 A & B
\end{bmatrix}^\tp
P_{t+1}
\begin{bmatrix}
 A & B
\end{bmatrix}, \\
\label{def:L_hat}
 \hat L_t^{nn} &= ( \begin{bmatrix}
 A & B
\end{bmatrix}_{n\bullet}) ^\tp
P_{t+1}^{nn}
\begin{bmatrix}
 A & B
\end{bmatrix}_{n\bullet}, \\ 
\tilde L_t^{nn} &= (\begin{bmatrix}
 A & B
\end{bmatrix}_{n \bullet}) ^\tp
\tilde P_{t+1}^{nn}
\begin{bmatrix}
 A & B
\end{bmatrix}_{n \bullet}.
\label{def:L_tilde}
\end{align}
The first equality in \eqref{eq:Vterm_6} is true because $S_t^{\theta_t}$ and $W_t^{1:N}$ are independent, and the second equality in \eqref{eq:Vterm_6} follows from the definition of $e_t$ in \eqref{eq:error}.

Using the random vector $S_t^{\theta_t}$, we can write the first term on the right hand side  of  \eqref{eq:DP_V} as

\begin{small}
\begin{align}
&\ee \left[ QF\left(R_t,S_t^{\theta_t} \right)\right] 
= QF\left(R_t,\ee \left[S_t^{\theta_t} \right]\right)
+ \tr\left(R_t\cov(S_t^{\theta_t}) \right).
\label{eq:Vterm_7}
\end{align}
\end{small}
 
Now putting  \eqref{eq:Vterm_7} and \eqref{eq:Vterm_6} (that is, the first and second terms of the right hand side of  \eqref{eq:DP_V}) together into the right hand side of  \eqref{eq:DP_V} we get

\begin{small}
\begin{align}
&e_t +  \min_{\{q_t^n\in\mathcal{Q}^{n}(\theta^n)\}_{n \in \mathcal{N}}}
\hspace{-2pt} \Big\lbrace \min_{\{\bar u^{n}_t\in\R^{d_U^n}\}_{n \in \mathcal{N}}, u^0_t\in\R^{d_U^0}}  \Big\lbrace
\nonumber\\
&\hspace{1cm} QF\left(G_t,\ee \left[S_t^{\theta_t} \right]\right)
+ \tr\left(\tilde G_t\cov(S_t^{\theta_t}) \right)
\Big\}\Big\},
\label{eq:DP_V_step1}
\end{align}
\end{small}
where we have defined 
\begin{align}
\label{def:G}
G_t &= R_t + L_t \\
\tilde G_t &= R_t + \sum_{n \in \mathcal{N}} \big( (1-p^n) \hat L_t^{nn} + p^n \tilde L_t^{nn} \big).
\label{def:G_tilde}
\end{align}

Note that $\ee [q_t^n(X^{\theta^n})] =0$ for all $n \in \mathcal{N}$, and consequently,
$\ee[S_t^{\theta_t}]= \vecc\Big(\{\mu(\theta_t^n)\}_{n \in \mathcal{N}}, u^0_t, \bar u^{1:N}_t\Big)$ depends only on $u^0_t,\bar u^{1:N}_t$.
Furthermore, $\cov(S_t^{\theta_t}) = \cov\left(\vecc(\{X^{\theta_t^{n}}\}_{n \in \mathcal{N}}, 0, \{q_t^n(X^{\theta_t^n}) \}_{n \in \mathcal{N}})\right)$ depends only on the choice of $q_t^{1:N}$. Consequently,
the optimization problem in the \eqref{eq:DP_V} can be further simplified to be

\vspace{-3mm}
\begin{small}
\begin{align}
  &e_{t} +
\min_{u^0_t,\bar u^{1:N}_t} 
QF\left(G_{t},\vecc(\{\mu(\theta_t^n)\}_{n \in \mathcal{N}}, u^0_t, \bar u^{1:N}_t) \right) +
\nonumber\\
&\min_{\{q_t^n\in\mathcal{Q}^{n}(\theta^n)\}_{n \in \mathcal{N}}}
\hspace{-5mm}
\tr\left(\tilde G_{t}  \cov\left(\vecc(\{X^{\theta_t^{n}}\}_{n \in \mathcal{N}}, 0, \{q_t^n(X^{\theta_t^n}) \}_{n \in \mathcal{N}})\right) \right).
\label{eq:DP_V_step3}
\end{align}
\end{small}

Now we need to solve the two optimization problems
\begin{small}
\begin{align}
 &\min_{u^0_t,\bar u^{1:N}_t} 
QF\left(G_{t},\vecc(\{\mu(\theta_t^n)\}_{n \in \mathcal{N}}, u^0_t, \bar u^{1:N}_t)\right)   ,
\label{eq:DP_V_P1}
\\
& \min_{\{q_t^n\in\mathcal{Q}^{n}(\theta^n)\}_{n \in \mathcal{N}}}
\hspace{-5mm}
\tr\left(\tilde G_{t}  \cov\left(\vecc(\{X^{\theta_t^{n}}\}_{n \in \mathcal{N}}, 0, \{q_t^n(X^{\theta_t^n}) \}_{n \in \mathcal{N}})\right) \right).
\label{eq:DP_V_P2}
\end{align}
\end{small}
We first consider the optimization in \eqref{eq:DP_V_P1}. According to \eqref{def:G} and \eqref{def:L}, we have
\begin{align}
G_t &= \left[\begin{array}{ll}
G^{XX}_t & G^{XU}_t  \\
G^{UX}_t & G^{UU}_t
\end{array}\right]   \notag \\
&=\left[\begin{array}{ll}
R^{XX}_t & R^{XU}_t   \\
R^{UX}_t & R^{UU}_t 
\end{array}\right] +
\left[\begin{array}{l}
A^{\tp} \\
B^{\tp}
\end{array}\right]
P_{t+1}
\begin{bmatrix}
A & B
\end{bmatrix}.
\label{def:G_details}
\end{align}
Since $R_t^{UU}$ is PD and further, $P_{t+1}$ is PSD by induction, $G_{t}^{UU}$ is PD. Then, it follows by the first part of Claim \ref{lm:quadratic_problems} that the optimal solution of \eqref{eq:DP_V_P1} is given by \eqref{eq:opt_ubar} and

\vspace{-3mm}
\begin{small}
\begin{align}
\label{eq:DP_V_part1}
&\min_{u^0_t,\bar u^{1:N}_t} 
QF\left(G_{t},\vecc(\{\mu(\theta_t^n)\}_{n \in \mathcal{N}}, u^0_t, \bar u^{1:N}_t)\right)    \notag \\
&=QF\Big(P_t, \vecc(\{\mu(\theta_t^n) \}_{n \in \mathcal{N}}) \Big),
\end{align}
\end{small}
where $P_t=G^{XX}_t - G^{XU}_t  (G^{UU}_t)^{-1} G^{UX}_t$ and $K_t = - (G^{UU}_t)^{-1} G^{UX}_t$.
From \eqref{def:G_details}, it is straightforward to see that $P_t =  \Omega (P_{t+1},A,B,R_t^{XX},R_t^{UU}, R_t^{XU})$ and $K_t = \Psi (P_{t+1},A,B,R_t^{UU}, R_t^{XU})$. Furthermore, since $P_{t+1}$ is PSD, according to \eqref{def:G_details}, $G_t$ is PSD. Then $P_t$, and consequently $P_t^{nn}$  for all $n \in \mathcal{N}$, are PSD because $P_t$ is the Schur complement of $G_t^{UU}$ of the matrix $G_t$.

Now, we consider the optimization in \eqref{eq:DP_V_P2}. 
We define the matrix $\tilde G_t^n$ as follows

\vspace{-3mm}
\begin{small}
\begin{align}
\tilde G_t^n &:= \left[\begin{array}{ll}
\tilde G^{X^nX^n}_t & \tilde G^{X^nU^n}_t  \\
\tilde G^{U^nX^n}_t & \tilde G^{U^nU^n}_t
\end{array}\right]  = 
 \left[\begin{array}{ll}
R^{X^nX^n}_t & R^{X^nU^n}_t  \\
R^{U^nX^n}_t & R^{U^nU^n}_t
\end{array}\right]   \notag \\
&+ \left[\begin{array}{l}
(A^{nn})^{\tp} \\
(B^{nn})^{\tp}
\end{array}\right]
\big((1-p^n) P_{t+1}^{nn} + p^n \tilde P_{t+1}^{nn}\big) 
\begin{bmatrix}
A^{nn} & B^{nn}
\end{bmatrix}
\label{def:G_tilde_details}
\end{align}
\end{small}
where the last equality is true from \eqref{def:L_hat}, \eqref{def:L_tilde}, and \eqref{def:G_tilde}.
Since $R_t^{U^nU^n}$ is PD and further, $P_{t+1}^{nn}$ and $\tilde P_{t+1}^{nn}$ are PSD by induction, $\tilde G_{t}^{U^nU^n}$ is PD. Then, the second part of Claim \ref{lm:quadratic_problems} implies that the optimal solution of \eqref{eq:DP_V_P2} is given by \eqref{eq:opt_gamma} and

\vspace{-3mm}
\begin{small}
\begin{align}
 &\min_{\{q_t^n\in\mathcal{Q}^{n}(\theta^n)\}_{n \in \mathcal{N}}}
\hspace{-5mm}
\tr\left(\tilde G_{t}  \cov\left(\vecc(\{X^{\theta_t^{n}}\}_{n \in \mathcal{N}}, 0, \{q_t^n(X^{\theta_t^n}) \}_{n \in \mathcal{N}})\right) \right)
\nonumber\\
&= \sum_{n=1}^N \tr\left(\tilde P_{t}^{nn} 
\cov \left(\theta_t^n\right)
\right),
\label{eq:DP_V_part2}
\end{align}
\end{small}
where $\tilde P_t^{nn}= \tilde G^{X^nX^n}_t - \tilde G^{X^nU^n}_t  (\tilde G^{U^nU^n}_t)^{-1} \tilde G^{U^nX^n}_t$ and $\tilde K_t^{nn} = - (\tilde G^{U^nU^n}_t)^{-1} \tilde G^{U^nX^n}_t$. From \eqref{def:G_tilde_details}, it is straightforward to see that $\tilde P_t^{nn} =  \Omega \big((1-p^n) P_{t+1}^{nn} + p^n \tilde P_{t+1}^{nn},A^{nn},B^{nn},R_t^{X^nX^n},R_t^{U^nU^n}, R_t^{X^nU^n} \big)$ and $\tilde K_t^{nn} = \Psi \Big((1-p^n) P_{t+1}^{nn} + p^n \tilde P_{t+1}^{nn},A^{nn},B^{nn},R_t^{U^nU^n}, R_t^{X^nU^n} \Big)$. Furthermore, since $\tilde P_{t+1}^{nn}$ is PSD, according to \eqref{def:G_tilde_details} $\tilde G_t^{n}$ is PSD. Then, $\tilde P_t^{nn}$  is PSD because $\tilde P_t^{nn}$ is the Schur complement of
$\tilde G^{U^nU^n}_t$ of the matrix $\tilde G_t^{n}$.

Finally, substituting \eqref{eq:DP_V_part1} and \eqref{eq:DP_V_part2} into \eqref{eq:DP_V_step3} we obtain that $V_t$ defined by  \eqref{eq:Vt_PacketDrop} is equal to the right hand side of  \eqref{eq:DP_V}.
This completes the proof of the induction step and the proof of the theorem.

\vspace{-3mm}
\section{Proof of Theorem \ref{thm:opt_strategies}}
\label{Proof_Thm_optimal_strategies}
Let $\hat X_t^n$, $n \in \mathcal{N}$, be the estimate (conditional expectation) of $X_t^n$ based on the common information $H^0_t$. Then, for any realization of the marginal common belief $\theta_t^n$, $\hat x_t^n = \mu(\theta_t^n)$ for all $n\in\mathcal{N}$. This together with Theorems \ref{thm:structure} and \ref{thm:Sol_packetdrop} result in \eqref{eq:opt_ULbarUR} and \eqref{eq:opt_UL}.
To show \eqref{eq:estimator_0} and \eqref{eq:estimator_t}, note that 
at time $t=0$, for any $n \in \mathcal{N}$ and for any realization $h^0_t$ of $H^0_t$,
\begin{align}
&\hat x_0^n = \mu(\theta_0^n) = \int y \theta_0^n(dy) \notag \\
&= 
\Big \lbrace \begin{array}{ll}
\int y \pi_{X_0^n}(dy) = \mu(\pi_{X_0^n}) & \text{ if }z_{0}^n= \emptyset,\\
 \int y [\rho (x_0^n)](dy) =  
 x_0^n & \text{ if }z_{0}^n = x_0^n.
\end{array} 
\end{align}
Therefore, \eqref{eq:estimator_0} is true. Furthermore, at time $t+1$ and for any realization $h^0_{t+1}$ of $H^0_{t+1}$,
let $\theta_{t+1}^{1:N}$ be the corresponding common beliefs and $u^{prs*}_t = \phi^{prs*}_t(h^0_t)$, then
\begin{align*}
\hat x_{t+1}^n = \mu(\theta_{t+1}^n) = \int  y [\psi_t^n(\theta_t^n,u^{prs*}_t, z_{t+1}^n)] (dy).
\end{align*}
If $z_{t+1}^n = x_{t+1}^n$, then $\hat x_{t+1}^n = \int y [\rho (x_{t+1}^n)] (dy) = x_{t+1}^n.$
\\
If $z_{t+1}^n = \emptyset$, then,
\begin{align}
&\hat x_{t+1}^n = \int y [\psi_t^n(\theta_t^n,u^{prs*}_t,q_t^n, \emptyset)] (dy)
\notag\\
= &\int  \int \int y [\rho(f^n_t(x_{t}^n, w_{t}^n,u^{prs*}_t))] (dy)
\theta_t^n(dx_t^n) \pi_{W_t^n}(dw_t^n)
\notag \\
= &\int \int  f^n_t(x_{t}^n, w_{t}^n,u^{prs*}_t)
\theta_t^n(dx_t^n) \pi_{W_t^n}(dw_t^n)
\notag\\
= &A^{nn} \hat x_t^n + B^{nn} \bar u_t^{n*} + B^{n0} u^{0*}_t.
\end{align}
where the third equality is true because 
\begin{align*}
&\int  y [\rho(f^n_t(x_{t}^n, w_{t}^n,u^{prs*}_t))](dy)
= f^n_t(x_{t}^n, w_{t}^n,u^{prs*}_t).
\end{align*}
Furthermore, the last equality is true because $q_t^n \in \mathcal{Q}^n(\theta^n)$ and $W_t^n$ is a zero mean random vector.
Therefore, \eqref{eq:estimator_t} is true and the proof is complete.